\pdfoutput=1
\documentclass[12pt]{article}

\usepackage[margin=1in]{geometry}
\usepackage{amsmath,amssymb,amsfonts,amsthm}
\usepackage{bm}
\usepackage{graphicx}
\graphicspath{{./}}
\usepackage[dvipsnames]{xcolor}
\usepackage{microtype}
\usepackage[font=small,labelfont=bf]{caption}
\usepackage[section]{placeins}
\usepackage{enumitem}
\usepackage[colorlinks=true,linkcolor=MidnightBlue,citecolor=MidnightBlue,urlcolor=MidnightBlue]{hyperref}
\usepackage{cite}
\newtheorem{theorem}{Theorem}
\newtheorem{remark}[theorem]{Remark}
\newtheorem{proposition}[theorem]{Proposition}
\title{Exact Schwarzian Metric Factor and \\
Holographic Wilson-Loop Screening}

\author{Miguel Tierz\\[4pt]
{\small Shanghai Institute for Mathematics and Interdisciplinary Sciences (SIMIS),}\\
{\small Block A, No.\ 657 Songhu Road, Yangpu District, Shanghai 200438, China}\\[2pt]
{\small \texttt{tierz@simis.cn}}}
\date{\today}
\begin{document}
\maketitle

\begin{abstract}
We evaluate exactly the radial metric factor $h(\zeta)$ generated by
Schwarzian averaging in the AdS$_2$ throat of an extremal
Reissner--Nordstr\"om AdS$_5$ black brane.
The result is a Gaussian integral against $\coth(\pi y)$, valid at all radial
depths, which Mordell's identity turns into an exact Appell--Lerch
$q/q^\ast$-series representation.
The dual series identifies the
nonperturbative scale $e^{-\pi^2 C/\zeta}$ missed by any finite
near-boundary truncation.
The third parametric derivative required by the evaluation generates the
quasimodular Eisenstein series $E_2$, absent from the classical Mordell identity.
From the integral representation we prove that
$\mathcal{G}_0(\zeta):=h(\zeta)/\zeta^2$ is completely monotone and hence has no
interior minimum, so any confining minimum produced by a finite near-boundary
truncation is an artifact.
We also compute the exact relative variance of the Schwarzian kernel,
which makes the averaged-metric approximation error quantitative and shows that
the absence of a confining minimum is robust across moment-based effective
geometries.
Applied to the temporal rectangular Wilson loop, the exact throat gives
algebraic screening, $E(L)\sim -\kappa_{\rm IR}/L^2$, the Wilson-loop
diagnostic of the semi-local quantum-liquid IR of the extremal RN brane.
A numerical check in a simple matched geometry
confirms that the screened saddle is the dominant string configuration,
and an exact-versus-truncated force comparison shows that the apparent
constant-force regime of the fourth-order truncation is not a feature of the
exact geometry.
\end{abstract}

\tableofcontents
\section{Introduction}

In the AdS/CFT correspondence~\cite{Maldacena:1997re,Witten:1998qj}, the heavy quark-antiquark potential at strong coupling is computed from the area of a string worldsheet ending on a Wilson loop at the AdS boundary~\cite{Maldacena:1998im,Rey:1998ik}. In conformal backgrounds the result is purely Coulombic~\cite{Maldacena:1998im,Rey:1998ik}. Turning on a finite charge density changes this picture: near-extremal Reissner--Nordstr\"om (RN) AdS black branes~\cite{Chamblin:1999tk} develop a near-horizon AdS$_2$ factor.  On the field-theory side this AdS$_2\times\mathbb{R}^3$ region realizes a semi-local quantum liquid, or local quantum critical IR sector, with time scaling but no spatial scaling ($z\to\infty$)~\cite{Faulkner:2009wj,Iqbal:2011ae}. The planar spatial directions are compactified on $T^3$ for the dimensional reduction below.  The low-energy quantum fluctuations in this region are controlled by Jackiw--Teitelboim (JT) gravity and its Schwarzian boundary mode~\cite{Jackiw:1984,Teitelboim:1983ux,Almheiri:2014cka,Faulkner:2009wj,Maldacena:2016hyu,Jensen:2016pah,Maldacena:2016upp,Yang:2018qgravity,Saad:2019lba} (see~\cite{Mertens:2022review} for a review). The Schwarzian effective action governs the pattern of conformal symmetry breaking in nearly AdS$_2$ spacetimes~\cite{Maldacena:2016hyu,Maldacena:2016upp,Kitaev:2017awl}, and its path integral can be evaluated exactly~\cite{Stanford:2017thb,Mertens:2017,Mertens:2018fds,Yang:2018qgravity,Saad:2019lba}.

This framework has led to a broad program of re-examining classical holographic results with the inclusion of quantum corrections from the Schwarzian sector. Quantum averaging over the Schwarzian mode~\cite{Blommaert:2019hjr,Iliesiu:2020qvm} produces nontrivial radial kernels in the AdS$_2$ throat. The same Schwarzian kernel, or closely related near-AdS$_2$ quantum factors, enters a range of holographic observables, including Wilson loops~\cite{Liu:2024}, transport coefficients~\cite{Liu:2024strange}, hydrodynamic modes~\cite{Nian:2025fluid}, the shear viscosity to entropy density ratio~\cite{PandoZayas:2025etaS,Cremonini:2025etaS}, Hawking radiation spectra~\cite{Maulik:2025hawking}, and quasi-normal mode frequencies~\cite{Jiang:2025qnm}. A common feature of these applications is that near-boundary or low-temperature expansions are often used beyond their manifest domain of validity, leading in some cases to qualitative modifications: a linear quark-antiquark potential~\cite{Liu:2024}, regularization of the fluid/gravity correspondence~\cite{Nian:2025fluid}, and nontrivial low-temperature corrections to $\eta/s$, with regime-dependent conclusions about the Kovtun--Son--Starinets bound~\cite{PandoZayas:2025etaS,Cremonini:2025etaS,Gouteraux:2025shear,Kanargias:2025shear}.

Although the Schwarzian two-point function is exactly known~\cite{Stanford:2017thb,Mertens:2017,Saad:2019lba}, the induced radial profile $h(\zeta)$, with $\zeta$ the AdS$_2$ radial coordinate (small near the boundary, large deep in the throat), has not previously been put in a form valid at all radial depths, as needed for analyzing holographic observables throughout the throat. In practical applications, the Schwarzian-corrected profile is represented by a near-boundary expansion in $\zeta/C$, truncated at finite order, where $C$ is the Schwarzian coupling setting the scale that separates the classical and quantum regimes. This is reliable close to the AdS$_2$ boundary but not deep in the throat, where $\zeta/C \gg 1$ and the asymptotic series loses uniform control. The issue is especially sharp for extended probes, whose saddles can penetrate to radial depths $\zeta \gtrsim C$. Wilson loops provide a useful diagnostic: one such finite-order implementation produced an approximately linear quark-antiquark potential at intermediate distances~\cite{Liu:2024}, reminiscent of the Cornell potential~\cite{Eichten:1975,Eichten:1978,Bali:2000gf} and raising the question of holographic confinement~\cite{Witten:1998zw,Kinar:1999,Sonnenschein:2000qm}. Whether this linear behavior persists in the exact geometry is a nonperturbative question that cannot be settled without the full radial profile.

In this paper we obtain the exact Schwarzian-corrected metric factor $h(\zeta)$ throughout the AdS$_2$ throat. The Schwarzian kernel reduces, after a change of variables, to a Gaussian-weighted integral against $\coth(\pi y)$. This exact integral has two uses:
\begin{enumerate}[itemsep=4pt,topsep=4pt]
\item Endpoint analysis directly yields the large-$\zeta$ behavior $h(\zeta)\sim(\zeta/C)^{1/2}$, rather than the polynomial growth implied by any finite truncation. A turning-point analysis in the near-horizon geometry then gives
\begin{equation}
  E(L)\;\sim\; -\frac{\kappa_{\rm IR}}{L^2}\,,\qquad L\gg C\,,
\end{equation}
for the renormalized quark-antiquark potential at large separation, with the coefficient $\kappa_{\rm IR}$ obtained in closed form (Eq.~\eqref{eq:kappa-IR-restored}). The potential exhibits \emph{algebraic screening} --- power-law decay rather than the exponential (Debye) screening of a finite-temperature plasma or the linear confinement of the truncated series.  This power law is the Wilson-loop manifestation of the scale-free semi-local IR of the extremal brane.  In Sec.~\ref{subsec:numerical-branch-selection} we also package the isolated throat result as a normalized universal curve, Eq.~\eqref{eq:universal-throat-function}, so that the IR answer can be reused independently of the UV completion.
\item The integral is of the type classified by Mordell~\cite{Mordell:1933}, and the Mordell identity provides an explicit closed-form evaluation of $h(\zeta)$ (Eq.~\eqref{eq:h-q-series}) in terms of Appell--Lerch sums and theta functions~\cite{Zwegers:2002,Andrews:2018} --- the $q$-series defined in Sec.~\ref{sec:mordell-eval} --- with two complementary $q$-expansions, one suited to the near-boundary regime and one to the deep interior. Carrying out the required third derivative additionally generates the weight-2 Eisenstein series $E_2$~\cite{Andrews:2018}, which mixes the two channels and is not itself part of the classical Mordell identity.
\end{enumerate}
The strict linear confinement and putative mass gap found in finite-order analyses are not properties of the Schwarzian throat itself but of its truncation: the ratio $\mathcal{G}_0(\zeta):=h(\zeta)/\zeta^2$, which controls the effective string tension in the throat, decays monotonically rather than developing a minimum, and the tension vanishes deep in the throat.

The integral representation~\eqref{eq:h-exact-normalized} directly yields the UV coefficients and complete monotonicity, and, by endpoint analysis, the IR $(\zeta/C)^{1/2}$ asymptotics. The Mordell/Appell--Lerch evaluation (Sec.~\ref{sec:mordell-eval}) goes further: by carrying out the required third parametric derivative of the Mordell identity and exploiting a cancellation between the direct and $S$-dual Appell--Lerch channels at the degenerate point, we obtain the closed-form $q/q^\ast$-series~\eqref{eq:h-q-series}. Any finite Taylor truncation misses the entire $S$-dual sector, controlled by $e^{-\pi^2 C/\zeta}$; the Mordell formula makes these nonperturbative corrections quantitatively accessible.\footnote{Here and below, ``nonperturbative'' refers to contributions of order $\exp(-\pi^2 C/\zeta)$ that are invisible to any finite asymptotic expansion in $\zeta/C$ of the zero-temperature Schwarzian kernel. This should not be confused with a nonperturbative completion of JT gravity in the matrix-integral or topology-summed sense~\cite{Saad:2019lba}.} The finite-order Wilson-loop analysis of~\cite{Liu:2024} is recovered as the near-boundary approximation to the exact result.

Our exact representation of $h(\zeta)$ holds at all radial depths of the AdS$_2$ throat. By itself it does not determine the UV gluing to the asymptotically AdS$_5$ region, which is theory-dependent: the same Schwarzian zero mode dominates universally at $T\to 0$, but the way this description emerges from the UV completion differs between, e.g., RN-AdS$_5$, RN without cosmological constant, and near-extremal Kerr. Accordingly, the large-$L$ screening result depends only on the throat asymptotics and is robust, but intermediate-distance observables continue to depend on the matching prescription, which we take from~\cite{Liu:2024}.
We also include a numerical check
(Sec.~\ref{subsec:numerical-branch-selection}) in a simple matched
geometry.  This check is not a first-principles construction of the full
quantum-corrected interpolation, but it verifies that the exact-throat screened
saddle lies on a monotone, concave branch in the natural RN-AdS$_5$
completion and is not pre-empted by an intermediate cusp or maximal-separation
transition.
All results are at strict extremality ($T=0$). The zero-temperature kernel is the $\beta\to\infty$ limit of the finite-temperature Schwarzian kernel (Sec.~\ref{sec:mordell-metric}), and finite-temperature extensions are discussed in Sec.~\ref{sec:outlook}.  The distinction is physical rather than technical: the infinitely long extremal throat supports the algebraic tail, while any nonzero temperature caps the throat and is expected to restore Debye-type exponential screening at sufficiently large separation.

The paper is organized as follows. Section~\ref{sec:geometry} reviews the quantum-corrected geometry. Section~\ref{sec:mordell-metric} derives the exact integral representation, its large-$\zeta$ asymptotics, the closed-form Mordell/Appell--Lerch evaluation of $h$, the general-$\Delta$ kernel family, and the variance analysis of the averaged-metric prescription. Section~\ref{sec:rectangular} contains the main Wilson-loop analysis: the large-distance screening, the role of truncation, the turning-point derivations, and the observable-level exact-versus-truncated force comparison. Section~\ref{sec:outlook} discusses open problems, and Appendix~\ref{app:numerics} records the numerical implementation details needed to reproduce the figures.
\section{Quantum-corrected near-horizon geometry}
\label{sec:geometry}

\medskip\noindent\textit{Conventions.}
Throughout we use the dimensionless ratio $\xi:=\zeta/C$ and the purely imaginary modular parameter $\tau:= i\xi/\pi$ (so $\Im\tau>0$). The nome (the standard exponential parameter of elliptic function theory) is $q:=e^{\pi i\tau}=e^{-\xi}$, and the image under the modular inversion $\tau\mapsto -1/\tau$ gives $q^\ast:=e^{-\pi i/\tau}=e^{-\pi^2/\xi}$. We follow the Jacobi $\theta_{11}$ conventions of~\cite{Andrews:2018}. All Schwarzian/JT computations (Sec.~\ref{sec:mordell-metric}) are carried out in Euclidean signature via the Wick rotation $t\to -it_E$. The Wilson-loop analysis (Sec.~\ref{sec:rectangular}) is likewise performed in Euclidean signature; the static quark-antiquark potential is then extracted from the large-$T$ behavior of the Euclidean Wilson loop in the standard way~\cite{Maldacena:1998im}.

We follow the RN-AdS$_5$ near-horizon setup and matching conventions of~\cite{Liu:2024}, while deriving below the exact form of the Schwarzian factor, valid at all radial depths. Consider a near-extremal Reissner--Nordstr\"om AdS$_5$ black brane with metric
\begin{equation}
  ds^2 = -\frac{u^2}{L_{\rm AdS}^2} f(u)\,dt^2
  + \frac{L_{\rm AdS}^2}{u^2}\frac{du^2}{f(u)}
  + \frac{u^2}{L_{\rm AdS}^2}d\vec{x}^{\,2},
  \label{eq:RN-metric}
\end{equation}
where $d\vec{x}^{\,2}$ is the flat metric on $T^3$ and
\begin{equation}
  f(u) = 1 - (1+Q^2)\frac{u_T^4}{u^4}
  + Q^2\frac{u_T^6}{u^6}.
\end{equation}
The extremal limit corresponds to $Q^2\to 2$ at fixed horizon radius $u_T$. In this limit the near-horizon region develops an AdS$_2$ factor~\cite{Faulkner:2009wj,Almheiri:2014cka,Maldacena:2016upp}. Introducing the small parameter $\lambda \equiv (u-u_T)/u_T$ and zooming into the region $\lambda\ll 1$, one finds that the metric reduces to AdS$_2\times T^3$ up to subleading terms. A convenient radial coordinate for the AdS$_2$ factor is
\begin{equation}
  u = u_T + \frac{L_{\rm AdS}^2}{12\,\zeta},
  \label{eq:zetaUmap}
\end{equation}
in terms of which the near-horizon metric takes the form
\begin{equation}
  ds^2_{\rm NHR} \approx
  \frac{L_2^2}{\zeta^2}(-dt^2 + d\zeta^2) +
  \frac{u_T^2}{L_{\rm AdS}^2}\,d\vec{x}^{\,2},
\end{equation}
with an effective AdS$_2$ radius $L_2$.

Dimensional reduction on $T^3$ and the near-horizon approximation lead to a JT gravity sector controlling quantum fluctuations of the AdS$_2$ factor~\cite{Almheiri:2014cka,Iliesiu:2020qvm,Maldacena:2016upp}. The Euclidean JT path integral over Schwarzian modes produces an exact boundary two-point function~\cite{Mertens:2017}. Physically, the Schwarzian mode is the boundary graviton of the nearly AdS$_2$ geometry: it parametrizes reparametrizations of the AdS$_2$ boundary, and the JT path integral averages over all such reparametrizations. This averaging modifies the two-point kernel that controls the bulk geometry. In the averaged-metric prescription of~\cite{Liu:2024}, this modification is encoded in an effective quantum-corrected metric, written here in Euclidean signature (the Wick rotation $t\to -it_E$ of the Lorentzian near-horizon form above, as appropriate for the Schwarzian path integral):
\begin{equation}
  \langle ds^2_{\rm AdS_2}\rangle
  = h(\zeta)\,\frac{L_2^2}{\zeta^2}(dt^2 + d\zeta^2),
  \label{eq:quantum-ads2}
\end{equation}
where the dimensionless metric factor $h(\zeta)$ encodes all quantum corrections from the Schwarzian sector. Its normalization is fixed by requiring that $h(\zeta)\to1$ as $\zeta\to 0$, where the geometry approaches the AdS$_2$ boundary.

The exact expression for $h(\zeta)$ follows from the Schwarzian boundary two-point function~\cite{Mertens:2017,Engelsoy:2016,Stanford:2017thb,Blommaert:2018oro,Iliesiu:2019xuh,Yang:2018qgravity,Maldacena:2016upp}. The parameter $C$ introduced above plays the role of a JT heat capacity. In terms of it the quantum-averaged kernel is given by
\begin{equation}
  G_{\partial\partial}(2\zeta)
  = \frac{1}{2C}\int_0^\infty d\omega
  \,\sinh\!\bigl(2\pi\sqrt{2C\omega}\bigr)\,
  e^{-2\zeta\omega}\,
  \Gamma\!\left(1+i\sqrt{2C\omega}\right)^2
  \Gamma\!\left(1-i\sqrt{2C\omega}\right)^2.
  \label{eq:Gpp-original}
\end{equation}
Up to an overall normalization fixed by $h(0)=1$, one can identify $h(\zeta) \propto (\zeta/C)^2\,G_{\partial\partial}(2\zeta)$. Expanding this expression near the boundary gives
\begin{equation}
  h(\zeta) \simeq 1 + c_2\left(\frac{\zeta}{C}\right)^2
  - c_3\left(\frac{\zeta}{C}\right)^3
  + c_4\left(\frac{\zeta}{C}\right)^4
  + \cdots,
  \label{eq:h-small-zeta}
\end{equation}
with the approximate coefficients $c_2=3/(2\pi^4)$, $c_3=15/(2\pi^6)$, $c_4=315/(8\pi^8)$. The exact coefficients, which differ by Riemann zeta factors, are derived in Sec.~\ref{sec:mordell-metric} once the exact integral representation is available.

In the far region the geometry is still described by the RN-AdS$_5$ metric \eqref{eq:RN-metric}, now dressed~\cite{Liu:2024} by the same quantum factor $h(u)$ that appears in the near-horizon AdS$_2$ patch. Following~\cite{Liu:2024}, we adopt a continuation/matching ansatz that extends the exact throat factor $h(\zeta)$ into the far AdS$_5$ region, giving the quantum-corrected AdS$_5$ metric in the coordinates $U \equiv u/\alpha'$, $U_T \equiv u_T/\alpha'$ as
\begin{equation}
  ds^2 = \alpha'
  \left[
  \frac{U^2}{R^2}\bigl(f(U)h(U)\,dt^2 + d\vec{x}^{\,2}\bigr)
  + \frac{R^2}{U^2}\,f(U)^{-1}h(U)\,dU^2
  \right],
  \label{eq:quantum-ads5}
\end{equation}
where $R$ is related to $L_{\rm AdS}$ and the effective string tension through $\sqrt{\alpha'} = L_{\rm AdS}/R$. The function $h(U)$ is the continuation of $h(\zeta)$ from the AdS$_2$ region along the radial flow. Near the boundary the expansion \eqref{eq:h-small-zeta} translates into an expansion of $h(U)$ in inverse powers of $(U-U_T)$. From the point of view of Wilson loops, all observables depend on $h(\zeta)$ only through the parametric integrals for the separation $L$ and the Nambu--Goto action $S_{\rm NG}$. The exact integral representation of $h(\zeta)$ is directly applicable to the large-$L$ regime, where the worldsheet probes the AdS$_2$ throat. At intermediate and short distances, the Wilson loop observables depend on $h(U)$ in the far region where the near-horizon and far metrics are matched; in that regime the results rely on the matching ansatz of~\cite{Liu:2024}. The geometry and the string configurations are illustrated schematically in Fig.~\ref{fig:schematic}.

\begin{figure}[t]
  \centering
  \includegraphics[width=0.48\linewidth]{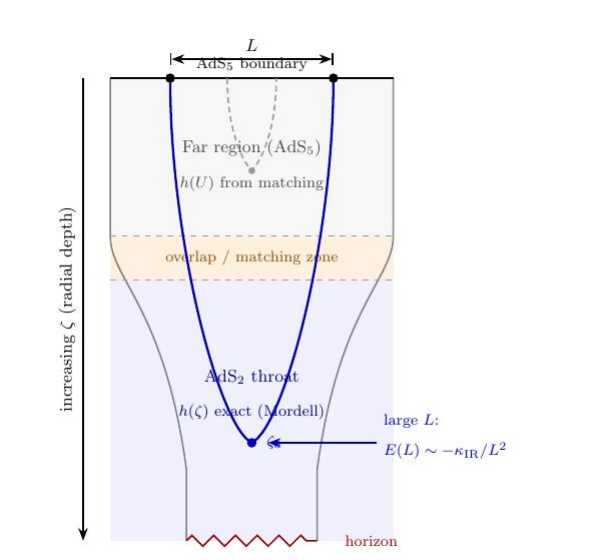}
  \caption{Schematic of the radial geometry and string worldsheets. The AdS$_5$ boundary is at the top; radial depth $\zeta$ increases downward toward the horizon. The near-horizon AdS$_2$ throat (blue shading) is where the quantum metric factor $h(\zeta)$ is controlled by the exact integral representation. The far region (gray) is the asymptotically AdS$_5$ part, where $h(U)$ is determined by matching. For large quark-antiquark separation $L$ (solid curve), the string turning point $\zeta_0$ lies deep in the throat, and the potential is governed by the exact $h(\zeta)\sim\sqrt{\zeta}$ scaling. For small $L$ (dashed curve), the turning point remains in the far region, and the results depend on the matched metric.}
  \label{fig:schematic}
\end{figure}

\section{Exact integral representation and Mordell structure of the metric factor}
\label{sec:mordell-metric}

In this section we reduce the Schwarzian kernel~\eqref{eq:Gpp-original} to a one-dimensional integral of the type classified by Mordell~\cite{Mordell:1933}. Throughout this section we work in the dimensionless radial variable $\xi:=\zeta/C$ introduced above, writing $h(\xi)$ for the same metric factor as $h(\zeta)$. Such integrals have appeared previously in the exact solution of Chern--Simons matrix models~\cite{Russo:2015cst} and in related JT gravity contexts~\cite{Blommaert:2018oro,Griguolo:2021bilocal}. An important distinction is worth noting at the outset: in the Chern--Simons setting of~\cite{Russo:2015cst}, the nome $q$ is a root of unity ($|q|=1$) and the Mordell integral evaluates to finite sums involving Gauss-type exponential sums, whereas in the present JT gravity application $q=e^{-\zeta/C}$ is real and lies strictly inside the unit disk. The Mordell identity then produces convergent $q$-series rather than finite sums, and the modular properties take a qualitatively different form. We extract the exact Gaussian-coth integral, its large-$\xi$ asymptotics, and, by carrying out the $\partial_x^3$ differentiation of the Mordell identity, the explicit closed-form evaluation of $h(\xi)$ in terms of Appell--Lerch sums, the Jacobi theta function, and the weight-2 Eisenstein series~$E_2$.  We then record the general-conformal-dimension extension and use it to quantify the variance of the averaged-metric prescription.

\subsection{From the Schwarzian kernel to an exact Gaussian-coth integral}

The four gamma functions in \eqref{eq:Gpp-original} can be grouped into two conjugate pairs. Using
\begin{equation}
\Gamma(1+iy)\,\Gamma(1-iy)=\frac{\pi y}{\sinh(\pi y)}\,,
\label{eq:gamma_pair}
\end{equation}
the product of the four gamma functions becomes $[\pi y/\sinh(\pi y)]^2$, while the $\sinh(2\pi y)$ factor from~\eqref{eq:Gpp-original} combines with this to give
\begin{equation}
\sinh(2\pi y)\times\left[\frac{\pi y}{\sinh(\pi y)}\right]^2=2\pi^2 y^2\,\coth(\pi y)\,.
\end{equation}
Substituting $y=\sqrt{2C\omega}$ (so $d\omega=y\,dy/C$ and $\omega=y^2/(2C)$) into~\eqref{eq:Gpp-original}, the factor $2\pi^2 y^2$ together with $d\omega = y\,dy/C$ and the prefactor $1/(2C)$ yields
\begin{equation}
G_{\partial\partial}(2\zeta)=\frac{\pi^2}{C^2}\int_0^\infty dy\,y^3\,e^{-(\zeta/C)y^2}\,\coth(\pi y)\,.
\label{eq:Gsingle}
\end{equation}
The quantum metric factor is identified with the finite boundary limit of $(\zeta/C)^2\,C^2\,G_{\partial\partial}(2\zeta)/\pi^2$; the explicit $(\zeta/C)^2$ prefactor is required because $G_{\partial\partial}$ diverges as $\zeta^{-2}$ near the boundary, while $h(\zeta)\to 1$. Thus
\begin{equation}
h(\zeta)=N_h\,\Bigl(\frac{\zeta}{C}\Bigr)^{2}\int_0^\infty dy\,y^3\,e^{-(\zeta/C)y^2}\,\coth(\pi y)\,,
\label{eq:hsingle}
\end{equation}
where $N_h$ is now fixed entirely by the boundary condition $h(0)=1$.
To determine $N_h$, write $\coth(\pi y)=1+2/(e^{2\pi y}-1)$, so
\begin{equation}
\lim_{\xi\to 0}\xi^2\int_0^\infty dy\,y^3\,e^{-\xi y^2}\coth(\pi y)
=\int_0^\infty dy\,y^3\,e^{-\xi y^2}\Big|_{\text{leading}}
=\frac{\Gamma(2)}{2\xi^2}\;\xrightarrow{\;\xi^2\times\;}\;\frac{1}{2}\,,
\end{equation}
and $h(0)=1$ requires $N_h=2$. The exact metric factor is therefore
\begin{equation}
h(\xi)=2\xi^{2}\int_0^\infty dy\,y^3\,e^{-\xi y^2}\,\coth(\pi y)\,.
\label{eq:h-exact-normalized}
\end{equation}

\medskip\noindent\textit{Near-boundary asymptotic expansion and exact coefficients.}
\label{sec:exact-Taylor}
The exact near-boundary expansion of~\eqref{eq:h-exact-normalized}, in the sense of an
asymptotic expansion as $\xi\to0^+$, is
\begin{equation}
h(\xi)\sim 1+4\sum_{m=0}^{\infty}(-1)^m\,
\frac{\Gamma(2m+4)\,\zeta_{\rm R}(2m+4)}{m!\,(2\pi)^{2m+4}}\;\xi^{m+2}
=1+\frac{\xi^2}{60}-\frac{\xi^3}{126}+\frac{\xi^4}{240}-\frac{\xi^5}{396}+\cdots,
\label{eq:h-exact-Taylor}
\end{equation}
where $\zeta_{\rm R}$ denotes the Riemann zeta function.
The coefficients are exact, but the series should not be interpreted as a
convergent Taylor series at $\xi=0$.  The exact and approximate~\cite{Liu:2024}
coefficients are related by
\begin{equation}
c_k^{\rm exact} = \zeta_{\rm R}(2k)\;c_k^{\rm approx},
\qquad k=2,3,4,\ldots ,
\label{eq:coeff-relation}
\end{equation}
where $c_k$ denotes the coefficient of $\xi^k$.  Equivalently, in the summation
index of~\eqref{eq:h-exact-Taylor},
$c_{m+2}^{\rm exact}=\zeta_{\rm R}(2m+4)c_{m+2}^{\rm approx}$.
For low orders the $\zeta_{\rm R}$ factors are numerically close to one, so the discrepancy is mild; the approximate coefficients of~\cite{Liu:2024} correspond to omitting the $\zeta_{\rm R}(2k)$ factor in each term. The approximate coefficients were used in~\cite{Liu:2024} to compute Wilson loops to leading order in $1/C$.

The factorial growth of the coefficients makes the expansion asymptotic.  For
fixed truncation order $N$ one may write
\begin{equation}
h(\xi)=
1+
4\xi^2
\sum_{m=0}^{N}
\frac{(-\xi)^m}{m!}
\frac{\Gamma(2m+4)\zeta_{\rm R}(2m+4)}{(2\pi)^{2m+4}}
+R_N(\xi),
\end{equation}
with the simple bound
\begin{equation}
|R_N(\xi)|\le
4\xi^{N+3}
\frac{\Gamma(2N+6)\zeta_{\rm R}(2N+6)}
{(N+1)!(2\pi)^{2N+6}},
\qquad \xi>0 .
\end{equation}
The least term occurs at order $m_{\rm opt}\sim \pi^2/\xi$~\cite{Dingle:1973}, and the associated
nonperturbative scale is $e^{-\pi^2/\xi}$.  This is precisely the
$S$-transformed scale $q^\ast=e^{-\pi^2/\xi}$ appearing in the Mordell
description (Sec.~\ref{sec:mordell-eval}).  Equivalently, if
$c_{m+2}$ denotes the coefficient of $\xi^{m+2}$ in~\eqref{eq:h-exact-Taylor},
then
\begin{equation}
 c_{m+2}=(-1)^m
 \frac{2(m+1)\Gamma(m+5/2)}{\sqrt{\pi}\,\pi^{2m+4}}
 \zeta_{\rm R}(2m+4)
 \sim
 (-1)^m\frac{2}{\sqrt{\pi}\,\pi^4}
 m^{5/2}m!\,\pi^{-2m} .
\label{eq:uv-coeff-borel-growth}
\end{equation}
Thus the Borel transform of the UV series has radius $\pi^2$.  Using
$\zeta_{\rm R}(2m+4)=\sum_{k\ge1}k^{-2m-4}$, the natural singularity family is
\begin{equation}
        t=-\pi^2 k^2,
        \qquad k=1,2,\ldots,
\label{eq:borel-singularities}
\end{equation}
matching the poles of $\coth(\pi y)$ at $y=ik$.  The singularities lie on the
negative Borel axis, so the UV expansion is directionally Borel summable for
positive real $\xi$.  The $q^\ast$ channel should therefore be understood not as
a positive-axis ambiguity, but as the modular/resurgent completion data
controlling optimal truncation and analytic continuation across the negative
Borel direction.  Finite near-boundary truncations encode only this asymptotic
UV sector and cannot be extrapolated reliably into the deep throat; the
consequences for the Wilson-loop potential are discussed in
Sec.~\ref{subsec:finite-order-truncation}.

\medskip\noindent\textit{Range of validity of the near-boundary series.}
Before using the truncated expansion~\eqref{eq:h-small-zeta} inside Wilson-loop integrals, we quantify where it is reliable. We compare $h(\xi)$ from~\eqref{eq:h-exact-normalized}, evaluated by direct numerical quadrature, to its truncation through $\mathcal{O}(\xi^4)$. Figure~\ref{fig:h_series} shows the near-boundary agreement and the corresponding relative error, making precise the domain in which the Taylor expansion can be used for short-distance observables and indicating when the exact integral is required.

\begin{figure}[t]
  \centering
  \includegraphics[width=0.92\linewidth]{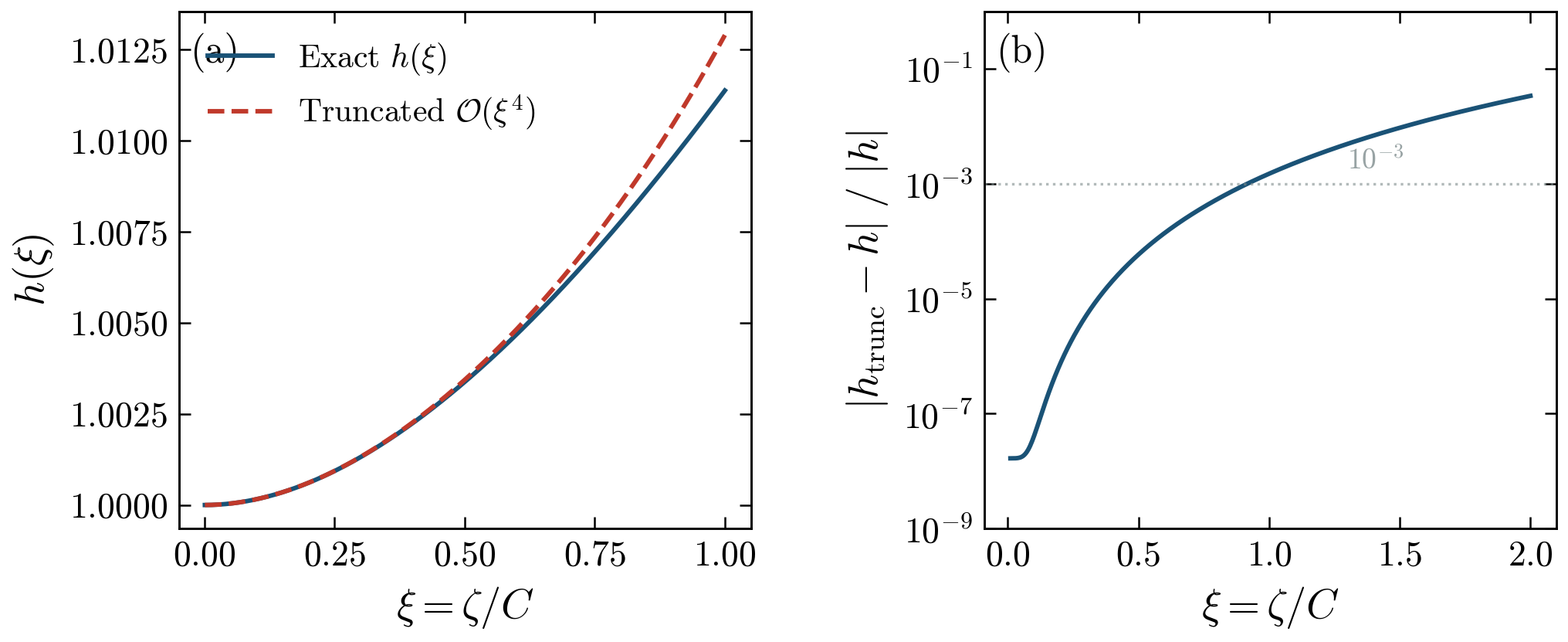}
  \caption{Near-boundary accuracy of the truncated expansion. (a) Exact quantum metric factor $h(\xi)$ (solid) vs.\ the $\mathcal{O}(\xi^4)$ truncation (dashed), with $\xi\equiv\zeta/C$. The two curves are indistinguishable for $\xi\lesssim 0.5$. (b) Relative error on a logarithmic scale. The error remains below $10^{-3}$ for $\xi\lesssim 0.9$ and grows rapidly beyond, motivating the exact integral representation for all infrared quantities.}
  \label{fig:h_series}
\end{figure}

Because $y^3\coth(\pi y)$ is even, one may symmetrize the integral and rewrite it as
\begin{equation}
\int_0^\infty dy\,y^3\,e^{-a y^2}\coth(\pi y)
=
\int_{-\infty}^{\infty} dt\,\frac{t^3\,e^{-a t^2}}{e^{2\pi t}-1}\,,
\qquad a=\frac{\zeta}{C}=\xi\,.
\label{eq:coth_to_bose}
\end{equation}
The right-hand side is a Gaussian-weighted moment against the kernel $1/(e^{2\pi t}-1)$, and integrals of this type fall within the class studied by Mordell~\cite{Mordell:1933}. The connection to Mordell's formalism and its modular properties is made precise in Sec.~\ref{sec:mordell-eval}.

\subsection{\texorpdfstring{Large-$\xi$ asymptotics from the coth integral}{Large-xi asymptotics from the coth integral}}
\label{sec:IR-asymp}
Before turning to the full Mordell evaluation, we extract the large-$\xi$ behavior of $h(\xi)$ directly from \eqref{eq:h-exact-normalized}. For $\xi\gg 1$ the Gaussian factor $e^{-\xi y^2}$ localizes the integral near $y=0$, where $\coth(\pi y)\simeq 1/(\pi y)$. The dominant contribution is therefore
\begin{equation}
\int_0^\infty dy\,y^3\,e^{-\xi y^2}\,\coth(\pi y)
\;\simeq\;
\frac{1}{\pi}\int_0^\infty dy\,y^2\,e^{-\xi y^2}
\;=\;
\frac{\sqrt{\pi}}{4\pi\,\xi^{3/2}}\,,
\qquad \xi\to\infty\,.
\label{eq:coth-saddle}
\end{equation}
Multiplying by $2\xi^2$ gives the leading behavior; including the next terms in the Laurent expansion $\coth(\pi y)=(\pi y)^{-1}+\pi y/3-\pi^3 y^3/45+\cdots$ and integrating term by term yields
\begin{equation}
h(\xi)\;\sim\;\frac{1}{2\sqrt{\pi}}\;\xi^{1/2}
+\frac{\pi^{3/2}}{4}\;\xi^{-1/2}
-\frac{\pi^{7/2}}{24}\;\xi^{-3/2}
+\cdots\,,\qquad \xi\to\infty\,.
\label{eq:h-IR-from-coth}
\end{equation}
The leading $\sqrt{\xi}$ coefficient $1/(2\sqrt{\pi})\approx 0.2821$ is the input for the long-distance Wilson loop analysis. The subleading powers of $1/\xi$ form a divergent asymptotic series. In addition, the Mordell identity discussed below implies the existence of exponentially suppressed corrections of order $e^{-\xi}$ from the $q$-expansion, which are invisible to the asymptotic power series. Figure~\ref{fig:h_full} compares the exact profile of $h(\xi)$ with the truncated series and the IR asymptotics~\eqref{eq:h-IR-from-coth} across the full range of depths.

\begin{figure}[t]
  \centering
  \includegraphics[width=0.65\linewidth]{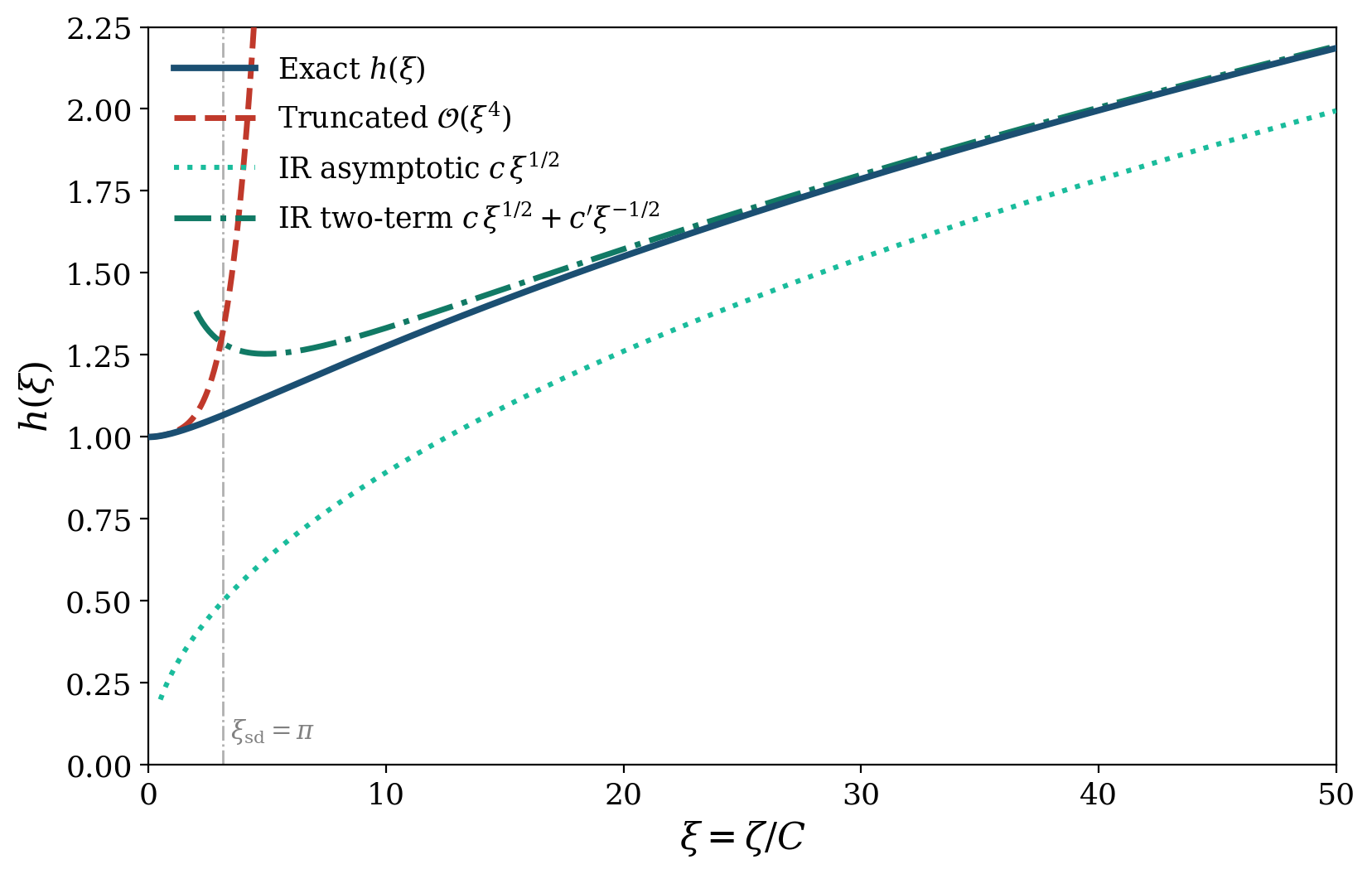}
  \caption{Global profile of the quantum metric factor $h(\xi)$ as a function of $\xi\equiv\zeta/C$ (solid), compared with the $\mathcal{O}(\xi^4)$ truncated series (dashed) and two IR asymptotic approximations from Eq.~\eqref{eq:h-IR-from-coth}: the leading term $c\,\xi^{1/2}$ with $c=1/(2\sqrt{\pi})$ (dotted), and the two-term approximation $c\,\xi^{1/2}+c'\xi^{-1/2}$ with $c'=\pi^{3/2}/4$ (dash-dotted). The truncation diverges for $\xi\gtrsim 2$. The leading $\sqrt{\xi}$ term undershoots by $\sim 9\%$ at $\xi=50$; including the subleading $\xi^{-1/2}$ correction reduces the error to $\sim 0.3\%$, confirming that the two-term asymptotic provides an accurate description of $h$ deep in the throat. The vertical dash-dotted line marks the self-dual point $\xi_{\rm sd}=\pi$ where the two modular expansion parameters $|q|=e^{-\xi}$ and $|q^\ast|=e^{-\pi^2/\xi}$ are equal.}
  \label{fig:h_full}
\end{figure}

\subsection{\texorpdfstring{Mordell identity and complementary $q$-channels}{Mordell identity and complementary q-channels}}
\label{sec:mordell-eval}

Introduce the Mordell integral~\cite{Mordell:1933} (in the normalization used in modern treatments)
\begin{equation}
\mathcal{M}(x,\theta;\tau)
:=
\int_{-\infty}^{\infty}
\frac{e^{\pi i\tau t^2-2\pi x t}}{e^{2\pi t}-e^{2\pi i\theta}}\;dt,
\qquad \Im\tau>0\,.
\label{eq:Mordell_def}
\end{equation}

Here $x$ is an auxiliary parameter in the sense of Mordell~\cite{Mordell:1933} that generates polynomial moments of $t$ by differentiation; it is
unrelated to the boundary spatial coordinate $x$ used later in the worldsheet embedding.

In our application we only need parametric derivatives of \eqref{eq:Mordell_def}. Indeed, differentiating under the integral sign shows that moments of $t$ are generated by $x$-derivatives:
\begin{equation}
\int_{-\infty}^{\infty} dt\;
\frac{t^3\,e^{\pi i\tau t^2}}{e^{2\pi t}-e^{2\pi i\theta}}
=
-\frac{1}{8\pi^3}\,
\left.\partial_x^{\,3}\mathcal{M}(x,\theta;\tau)\right|_{x=0}\,.
\label{eq:moment_dictionary}
\end{equation}
For $\theta=0$ the denominator in \eqref{eq:Mordell_def} vanishes at $t=0$, so the raw integral
$\mathcal{M}(x,0;\tau)$ is naturally understood by analytic continuation in $\theta$ (or, equivalently,
a principal-value contour prescription). In the present application the relevant object is
$\partial_x^{\,3}\mathcal{M}$ evaluated at $x=0$: the three $x$-derivatives bring down a factor of
$t^3$, and the would-be $1/t$ singularity at the origin becomes integrable. Consequently the moment
integral in \eqref{eq:moment_dictionary} is absolutely convergent and the $\theta\to 0$ limit can be taken
without ambiguity. More precisely, the Mordell identity~\eqref{eq:Mordell_identity} holds for generic $\theta$, and
after the $\partial_x^3$ differentiation the resulting integrands are uniformly bounded by a $\theta$-independent
$L^1$ function (provided by the Gaussian decay $e^{-\xi t^2}$), so the identity extends to $\theta\to 0$ by dominated convergence.

Combining \eqref{eq:h-exact-normalized} and \eqref{eq:coth_to_bose} with \eqref{eq:moment_dictionary} gives the relation between $h(\zeta)$ and the Mordell integral:
\begin{equation}
h(\zeta)
=
-\frac{1}{4\pi^3}\,\Bigl(\frac{\zeta}{C}\Bigr)^{2}\,
\left.
\partial_x^{\,3}\mathcal{M}(x,\theta;\tau)\right|_{x=0,\;\theta\to 0,\;\tau=i\zeta/(\pi C)}\,.
\label{eq:JT_to_Mordell_dictionary}
\end{equation}
Mordell~\cite{Mordell:1933} showed that \eqref{eq:Mordell_def} can be evaluated in closed form. Following the conventions of~\cite{Andrews:2018}, in the notation $q=e^{\pi i\tau}$ and $\Im\tau>0$, define
\begin{align}
i\,F(z,\tau)&:=\sum_{m=-\infty}^{\infty}\frac{(-1)^m\,q^{m^2+m+1/4}\,e^{(2m+1)\pi i z}}{1+q^{2m+1}}\,,
\label{eq:AppellLerch_def}\\[3pt]
i\,\theta_{11}(z,\tau)&:=\sum_{m=-\infty}^{\infty}(-1)^m\,q^{m^2+m+1/4}\,e^{(2m+1)\pi i z}\,.
\label{eq:theta11_def}
\end{align}
Here $\theta_{11}$ is the standard odd Jacobi theta function~\cite{Andrews:2018} and $F$ is an Appell--Lerch sum~\cite{Zwegers:2002} (the ratio $F/\theta_{11}$ is the holomorphic Appell--Lerch, or mock-Jacobi, object whose completed form has the modular behavior described by Zwegers~\cite{Zwegers:2002}). In terms of these functions, Mordell's identity reads~\cite{Mordell:1933}
\begin{equation}
\mathcal{M}(x,\theta;\tau)
=
e^{-\pi i(\theta^2\tau+2\theta x+2\theta)}\,
\frac{
F\!\left(\frac{x+\theta\tau}{\tau},-\frac{1}{\tau}\right)
+i\tau\,F(x+\theta\tau,\tau)
}{
\tau\,\theta_{11}(x+\theta\tau,\tau)
}\,.
\label{eq:Mordell_identity}
\end{equation}
The identity \eqref{eq:Mordell_identity} provides two complementary expansions: the term $F(x+\theta\tau,\tau)$ is a series in $q=e^{-\zeta/C}$ suited to large $\zeta$, while $F\bigl((x+\theta\tau)/\tau,-1/\tau\bigr)$ is an $S$-transformed series in $q^\ast=e^{-\pi^2 C/\zeta}$ suited to small $\zeta$.
We now carry out the $\partial_x^3$ differentiation and extract the explicit dual $q$-series for $h(\xi)$.

\medskip\noindent\textit{Holomorphicity at the degenerate point.}
Set $z=x+\theta\tau$, write $\tau_\ast:=-1/\tau$ for the modular-inverted parameter, and define
\begin{equation}
\Phi(z;\tau):=F\!\left(\frac{z}{\tau},\tau_\ast\right)+i\tau\,F(z,\tau)\,,
\qquad
R(z;\tau):=\frac{\Phi(z;\tau)}{\tau\,\theta_{11}(z,\tau)}\,.
\label{eq:Phi-R-def}
\end{equation}
The $S$-transformation identity for the Appell--Lerch sum implies that the
numerator vanishes at the same point where $\theta_{11}$ has its simple zero:
\begin{equation}
\Phi(0;\tau)=F(0,\tau_\ast)+i\tau\,F(0,\tau)=0\,.
\label{eq:Phi0-cancel}
\end{equation}
Since $\theta_{11}(z,\tau)$ has a simple zero at $z=0$, the cancellation~\eqref{eq:Phi0-cancel} implies
that $R(z;\tau)$ is holomorphic at $z=0$, despite the individual pole of
$\Phi/(\tau\theta_{11})$.

\medskip\noindent\textit{Extraction of $\partial_x^3$.}
From the Mordell identity~\eqref{eq:Mordell_identity},
\begin{equation}
\mathcal{M}(x,\theta;\tau)
=e^{-\pi i(\theta^2\tau+2\theta)}\,e^{-2\pi i\theta x}\,R(x+\theta\tau;\tau)\,.
\end{equation}
Since $R$ is holomorphic at the origin, the exponential prefactor
$e^{-2\pi i\theta x}$ contributes only subleading terms in $\theta$, and the
$\theta\to 0$ limit gives
\begin{equation}
\left.\partial_x^3\mathcal{M}(x,\theta;\tau)\right|_{x=0,\,\theta\to 0}
=R'''(0;\tau)\,.
\label{eq:d3M-Rppp}
\end{equation}
To compute $R'''(0)$, expand $\Phi$ and $\Theta\equiv\theta_{11}$ in Taylor series around $z=0$.
Write $\phi_j:=\tau^{-j}F^{(j)}(0,\tau_\ast)+i\tau\,F^{(j)}(0,\tau)$, so that
$\Phi(z)=\phi_1 z+\tfrac{1}{2}\phi_2 z^2+\tfrac{1}{6}\phi_3 z^3+\tfrac{1}{24}\phi_4 z^4+\cdots$\,.
Expanding $\Theta=\Theta_1 z+\tfrac{1}{6}\Theta_3 z^3+\cdots$ (only odd powers, since
$\theta_{11}$ is odd) and dividing, one obtains
\begin{equation}
R(z)=\frac{1}{\tau\Theta_1}\left[\phi_1+\frac{\phi_2}{2}z
+\left(\frac{\phi_3}{6}-\frac{\phi_1\Theta_3}{6\Theta_1}\right)z^2
+\left(\frac{\phi_4}{24}-\frac{\phi_2\Theta_3}{12\Theta_1}\right)z^3
+O(z^4)\right],
\label{eq:R-expansion}
\end{equation}
whence
\begin{equation}
R'''(0;\tau)=\frac{1}{\tau\Theta_1}
\left[\frac{\phi_4}{4}-\frac{\phi_2\Theta_3}{2\Theta_1}\right].
\label{eq:Rppp}
\end{equation}

\medskip\noindent\textit{Explicit $q/q^\ast$-series formula.}
Using the standard identity $\theta_{11}'''(0,\tau)/\theta_{11}'(0,\tau)=-\pi^2 E_2(\tau)$~\cite{Andrews:2018},
where $E_2(\tau)=1-24\sum_{n=1}^{\infty}\sigma_1(n)\,q^{2n}$ is the weight-2 Eisenstein series,
Eqs.~\eqref{eq:d3M-Rppp} and~\eqref{eq:Rppp} give
\begin{equation}
h(\xi)
=-\frac{\xi^2}{16\pi^3\,\tau\,\theta_{11}'(0,\tau)}
\Big[\phi_4+2\pi^2 E_2(\tau)\,\phi_2\Big]_{\,\tau=i\xi/\pi}\,,
\label{eq:h-Mordell-closed}
\end{equation}
with
\begin{align}
\phi_2&=\tau^{-2}F''(0,\tau_\ast)+i\tau\,F''(0,\tau)\,,\label{eq:phi2}\\
\phi_4&=\tau^{-4}F^{(4)}(0,\tau_\ast)+i\tau\,F^{(4)}(0,\tau)\,.\label{eq:phi4}
\end{align}

\begin{remark}
The final answer involves $F''$ and $F^{(4)}$, not $F'''$.
This is structural: the third $x$-derivative acts on a quotient whose denominator
vanishes linearly, so the $z^3$ coefficient of $R$ requires the $z^4$ coefficient
of $\Phi$.
\end{remark}

\medskip\noindent\textit{Explicit $q/q^\ast$-series form.}
For $\tau=i\xi/\pi$ we have $q=e^{-\xi}$, $q^\ast=e^{-\pi^2/\xi}$, and
$s:=\xi/\pi=-i\tau$.  Define
\begin{equation}
T(q):=2\sum_{m=0}^{\infty}(-1)^m(2m+1)\,q^{(m+1/2)^2}\,,
\label{eq:Tq-def}
\end{equation}
and, for $r=2,4$,
\begin{equation}
L_r(q):=\sum_{m=0}^{\infty}(-1)^m(2m+1)^r\,q^{(m+1/2)^2}
\,\frac{1-q^{2m+1}}{1+q^{2m+1}}\,.
\label{eq:Lr-def}
\end{equation}
Then $\theta_{11}'(0,\tau)=\pi T(q)$, and the Eisenstein series in the $q$-variable is
\begin{equation}
E_2(q):=1-24\sum_{n=1}^{\infty}\sigma_1(n)\,q^{2n}\,.
\label{eq:E2q-def}
\end{equation}
After reducing to real form,
\begin{equation}
h(\xi)=\frac{\pi^2 s}{16\,T(q)}
\left[
s^{-4}L_4(q^\ast)-s\,L_4(q)
+2E_2(q)\bigl(s^{-2}L_2(q^\ast)+s\,L_2(q)\bigr)
\right].
\label{eq:h-q-series}
\end{equation}
Equation~\eqref{eq:h-q-series} is our closed-form evaluation of $h(\xi)$, obtained by applying Mordell's identity to the exact integral~\eqref{eq:h-exact-normalized}.  The $q^\ast$-terms
($L_r(q^\ast)$) are the $S$-dual channel, exponentially small near the boundary
($\xi\ll 1$). The $q$-terms ($L_r(q)$) are exponentially small deep in the throat
($\xi\gg 1$).  The Eisenstein series $E_2$ mixes the two channels, reflecting the
quasi-modular (rather than strictly modular) nature of the weight-2 object.  We
emphasize that $E_2$ is not present in Mordell's classical
identity~\eqref{eq:Mordell_identity}, which involves only the Appell--Lerch and
theta functions.  It is generated here by the third derivative $\partial_x^3$,
which forces the third Taylor coefficient of the odd theta function $\theta_{11}$
at the origin and hence the ratio $\theta_{11}'''(0,\tau)/\theta_{11}'(0,\tau)=-\pi^2 E_2(\tau)$
[Eq.~\eqref{eq:h-Mordell-closed}].  The appearance of $E_2$ in the closed form
for $h$ is thus a feature of the present construction, not of the underlying
Mordell identity.

Direct summation of~\eqref{eq:h-q-series} reproduces the Gaussian-coth
integral to the full working precision of the numerical evaluation
($\approx 12$ significant figures). For instance, quoting twelve figures,
$h(0.5)=1.00337597688$, $h(3)=1.06251163839$, and
$h(10)=1.27670420355$.

Although the individual theta/Appell--Lerch sums $L_r(q^\ast)$ and $T(q)$ contain the standard
$q_\ast^{1/4}$-type theta prefactor, the same factor appears in $T(q)$ after modular
transformation and cancels in the ratio; the nonperturbative scale of the near-boundary
completion is therefore $q^\ast=e^{-\pi^2/\xi}$, in agreement with the optimal-truncation
estimate of Sec.~\ref{sec:exact-Taylor}.

The quasi-modular term has a direct interpretation: it is the curvature of the
$\theta_{11}$ zero that cancels the Appell--Lerch pole.  Thus the two channels are
not independent $q$- and $q^\ast$-series; they are mixed by the weight-two modular
connection~$E_2$.  This is the precise sense in which the exact throat geometry
is a quasi-modular, rather than strictly modular, completion of the near-boundary
expansion.

We note that the kernel $G_{\partial\partial}$ used throughout this paper is the strict $\beta\to\infty$ limit of the finite-temperature double-spectral representation of~\cite{Mertens:2017,Saad:2019lba}; thus the zero-temperature results form an exact subsector of the full thermal theory, and finite-temperature extensions correspond to controlled deformations of this limit rather than independent assumptions.

\subsection{General conformal dimension: ladders, Fermi towers, and infrared structure}
\label{sec:general-delta}

The same reduction applies to a Schwarzian bilocal of arbitrary conformal
dimension.  With the normalization chosen so that $h_\Delta(0)=1$, define
\begin{equation}
 h_\Delta(\xi)
 =\frac{\xi^{2\Delta}}{\pi^2\Gamma(2\Delta)}
 \int_0^\infty dy\; y\sinh(2\pi y)
 \left|\Gamma(\Delta+i y)\right|^{4}e^{-\xi y^2},
 \qquad \Delta>0 .
\label{eq:hDelta-def}
\end{equation}
For $\Delta=1$, the identity
$|\Gamma(1+i y)|^2=\pi y/\sinh(\pi y)$ reduces
\eqref{eq:hDelta-def} to the Gaussian-coth integral
\eqref{eq:h-exact-normalized}.  The recursion
\begin{equation}
 \left|\Gamma(\Delta+1+i y)\right|^{4}
 =\left(\Delta^2+y^2\right)^2
 \left|\Gamma(\Delta+i y)\right|^{4}
\label{eq:hDelta-ladder}
\end{equation}
turns each integer or half-integer tower into a finite differential ladder:
multiplication by $y^2$ inside the integral is $-\partial_\xi$, so the next rung
is obtained from the preceding seed by the operator
$(\Delta^2-\partial_\xi)^2$, followed only by the elementary change of the
normalizing factor $\xi^{2\Delta}/\Gamma(2\Delta)$.  Consequently any closed
$q/q^\ast$ representation for a seed propagates through the tower by
term-by-term differentiation.

For positive integers this gives the moment family
\begin{equation}
 h_n(\xi)
 :=\frac{\langle\mathcal{B}_1^{\,n}\rangle}
 {\langle\mathcal{B}_1\rangle_{\rm cl}^{\,n}}
 =\frac{2\,\xi^{2n}}{\Gamma(2n)}\int_0^\infty dy\;
 y^{3}\Bigl[\prod_{j=1}^{n-1}\bigl(j^{2}+y^{2}\bigr)\Bigr]^{2}
 e^{-\xi y^{2}}\,\coth(\pi y),
 \qquad h_n(0)=1 .
\label{eq:hn-def}
\end{equation}
For half-integers the seed is instead a Fermi-type kernel.  Since
$|\Gamma(1/2+i y)|^2=\pi/\cosh(\pi y)$,
\begin{equation}
 h_{1/2}(\xi)=2\xi\int_0^\infty dy\; y\tanh(\pi y)e^{-\xi y^2} .
\label{eq:hhalf-def}
\end{equation}
The near-boundary expansion follows by writing
$\tanh(\pi y)=1-2/(e^{2\pi y}+1)$; relative to the Bose integrals in
\eqref{eq:h-exact-Taylor}, the even zeta values are replaced by the Dirichlet
eta combination $(1-2^{1-2k})\zeta_{\rm R}(2k)$.  In particular
\begin{equation}
 h_{1/2}(\xi)=1-\frac{\xi}{12}+\frac{7\xi^2}{480}
 -\frac{31\xi^3}{8064}+O(\xi^4),
 \qquad \xi\to0^+ .
\label{eq:hhalf-UV}
\end{equation}
This is qualitatively different from the dimension-one kernel: the
half-integer seed initially decreases rather than grows.

The infrared endpoint is universal and is governed only by the small-$y$
expansion of the spectral density.  Expanding
$y\sinh(2\pi y)|\Gamma(\Delta+i y)|^4$ at fixed $\Delta$ gives
\begin{equation}
 h_\Delta(\xi)\simeq
 \frac{\Gamma(\Delta)^4}{\Gamma(2\Delta)}\,
 \frac{\xi^{2\Delta-3/2}}{2\sqrt{\pi}}
 \left[1+\frac{\pi^2-3\psi'(\Delta)}{\xi}+O(\xi^{-2})\right],
 \qquad \xi\to\infty .
\label{eq:hDelta-edge}
\end{equation}
Here $\psi'$ is the trigamma function.  At $\Delta=1$, where
$\psi'(1)=\pi^2/6$, Eq.~\eqref{eq:hDelta-edge} reduces to
\eqref{eq:h-IR-from-coth}.  The nonperturbative scales also have a simple
geography.  The poles of $|\Gamma(\Delta+i y)|^2$ at
$y=\pm i(\Delta+m)$ generate the deep-throat recessive scale
\begin{equation}
 q_\Delta=e^{-\Delta^2\xi},
 \qquad
 \xi_{\rm sd}(\Delta)=\frac{\pi}{\Delta},
\label{eq:qDelta-selfdual}
\end{equation}
where the second equation is the point at which $q_\Delta$ has the same size as
the UV modular scale $q^\ast=e^{-\pi^2/\xi}$.  In the integer tower the
polynomial zeros in \eqref{eq:hn-def} cancel the sub-$\Delta$ rungs of this pole
sequence, so the first surviving exponential is precisely $e^{-n^2\xi}$.

There is also a modular dichotomy.  Integer $\Delta$ sits at the degenerate
characteristic used above: the cancellation $\Phi(0;\tau)=0$ exposes the
quasi-modular connection term $E_2$.  Half-integer $\Delta$ instead sits at the
regular characteristic $\theta=1/2$; the corresponding Mordell integral belongs
to the family of mock theta functions classified by Zwegers~\cite{Zwegers:2002}
and transforms with a genuine weight-one-half $S$-factor before the differential
ladder is applied.  Thus the appearance of $E_2$ is a property of the
integer-dimension, degenerate-characteristic kernel, not an unavoidable feature
of every Schwarzian radial factor.

\subsection{Variance of the Schwarzian kernel and the status of the
averaged-metric prescription}
\label{sec:variance}

All Wilson-loop results in this paper, like those of~\cite{Liu:2024}, are
computed in the \emph{mean} geometry: the worldsheet functional is evaluated on
the quantum-averaged metric \eqref{eq:quantum-ads2}--\eqref{eq:quantum-ads5},
i.e.\ we compute $W[\langle g\rangle]$. The JT path integral itself defines
instead $\langle W[g]\rangle$, the Schwarzian average of the worldsheet
observable, and because the Nambu--Goto action is nonlinear in the metric the
two differ. Schematically, at fixed embedding,
\begin{equation}
\log\langle W\rangle
= -\,S_{\rm NG}[\langle g\rangle]
-\Bigl(\langle S_{\rm NG}[g]\rangle - S_{\rm NG}[\langle g\rangle]\Bigr)
+\tfrac12\,\mathrm{Var}\,S_{\rm NG}[g]\;+\;\cdots,
\label{eq:logW-cumulants}
\end{equation}
where both correction terms are quadratic in metric fluctuations at leading
order and are controlled by connected correlators of the Schwarzian bilocal at
split arguments, which are exactly known~\cite{Mertens:2017,Stanford:2017thb}
(re-extremization of the embedding contributes at the same order). Both
corrections raise $\log\langle W\rangle$: the variance is nonnegative, and
concavity of $\sqrt{\det g}$ on positive metrics gives
$\langle S_{\rm NG}[g]\rangle\le S_{\rm NG}[\langle g\rangle]$ at fixed
embedding, so to this order metric fluctuations \emph{deepen} the potential
rather than stiffen it; the mean-geometry estimate, if anything, underestimates
screening. The exposure to the difference is largest precisely where our main
result lives, deep in the throat. In this subsection we quantify it: the
variance of the very kernel that defines $h$ is exactly computable within the
same Gaussian-coth/Mordell toolbox, and it delimits where the mean geometry is
a controlled proxy for the fluctuating one.

\medskip\noindent\textit{Exact variance of the defining kernel.}
The kernel \eqref{eq:Gpp-original} is the zero-temperature expectation value of
the dimension-one Schwarzian bilocal
\begin{equation}
\mathcal{B}_\Delta(t_1,t_2)\;:=\;
\left[\frac{f'(t_1)\,f'(t_2)}{\bigl(f(t_1)-f(t_2)\bigr)^{2}}\right]^{\Delta},
\qquad
G_{\partial\partial}(2\zeta)\;\propto\;\bigl\langle\mathcal{B}_1(2\zeta)\bigr\rangle\,,
\label{eq:bilocal-functional}
\end{equation}
where $f$ is the boundary reparametrization and $\langle\,\cdot\,\rangle$ the
Schwarzian path integral. Geometrically $\mathcal{B}_1=e^{-\ell_{\rm ren}}$,
with $\ell_{\rm ren}$ the renormalized geodesic length between the two boundary
points, so that $h(\zeta)=\langle e^{-(\ell_{\rm ren}-\ell_{\rm cl})}\rangle$
is an \emph{annealed} average over length fluctuations. Because
$\mathcal{B}_\Delta$ is a functional of $f$ rather than an operator with
nontrivial ordering, its square at coincident endpoints is literally the
dimension-two bilocal, $\mathcal{B}_1^{\,2}=\mathcal{B}_2$, and the variance of
the defining kernel is the $\Delta=2$ member of the exactly solvable
family~\cite{Mertens:2017,Stanford:2017thb,Yang:2018qgravity}. Since
$|\Gamma(2+iy)|^{2}=(1+y^{2})\,|\Gamma(1+iy)|^{2}$, the reduction of
Sec.~\ref{sec:mordell-metric} applies verbatim with the polynomial
$y^{3}\mapsto y^{3}(1+y^{2})^{2}$, and the normalized second moment is
\begin{equation}
h_2(\xi)\;:=\;\frac{\langle\mathcal{B}_2(2\zeta)\rangle}
{\mathcal{B}_2^{\rm cl}(2\zeta)}
\;=\;\frac{\xi^{4}}{3}\int_0^\infty dy\,
\bigl(y^{3}+2y^{5}+y^{7}\bigr)\,e^{-\xi y^{2}}\,\coth(\pi y)\,,
\qquad h_2(0)=1\,.
\label{eq:h2-coth}
\end{equation}
The $y^{5}$ and $y^{7}$ moments are the derivative moments
$(-1)^{n}\mathcal{G}_0^{(n)}(\xi)=2\int_0^\infty dy\,y^{3+2n}e^{-\xi y^{2}}
\coth(\pi y)$ of the confinement indicator
$\mathcal{G}_0(\xi):=h(\xi)/\xi^{2}$ (cf.\
Theorem~\ref{thm:complete-monotone}), so the relative variance has the closed
form
\begin{equation}
\mathcal{V}(\xi)\;:=\;
\frac{\bigl\langle\mathcal{B}_1^{\,2}\bigr\rangle-\bigl\langle\mathcal{B}_1\bigr\rangle^{2}}
{\bigl\langle\mathcal{B}_1\bigr\rangle^{2}}
\;=\;\frac{h_2(\xi)}{h(\xi)^{2}}-1
\;=\;\frac{\mathcal{G}_0(\xi)-2\,\mathcal{G}_0'(\xi)+\mathcal{G}_0''(\xi)}
{6\,\mathcal{G}_0(\xi)^{2}}\;-\;1\,,
\label{eq:var-exact}
\end{equation}
an exact identity at all radial depths. Each moment is a higher $x$-derivative of the
Mordell integral at the degenerate point [cf.\
\eqref{eq:moment_dictionary}], so $\mathcal{V}$ inherits a closed
$q/q^{\ast}$-series representation of precisely the type
\eqref{eq:h-q-series}, with $\partial_x^{5}$ and $\partial_x^{7}$ replacing
$\partial_x^{3}$; we do not display it.

\medskip\noindent\textit{Asymptotics and crossover.}
Near the boundary the polynomial Gaussian moments in \eqref{eq:h2-coth} are
exact and terminate, $h_2(\xi)=1+\tfrac{2}{3}\xi+\tfrac{1}{6}\xi^{2}
+\mathcal{O}(\xi^{4})$, with corrections of order $\xi^{4}$ from the Bose tail
$\coth(\pi y)-1$, in the same asymptotic sense as \eqref{eq:h-exact-Taylor};
combining with \eqref{eq:h-exact-Taylor},
\begin{equation}
\mathcal{V}(\xi)\;=\;\frac{2}{3}\,\xi+\frac{2}{15}\,\xi^{2}
-\frac{2}{315}\,\xi^{3}+\mathcal{O}(\xi^{4})\,,
\qquad \xi\to0^{+}.
\label{eq:var-UV}
\end{equation}
Relative fluctuations of the kernel grow \emph{linearly} in $\zeta/C$ --- the
expected one-loop behavior of the weakly coupled Schwarzian mode, here with
exact coefficient $2/3$. Deep in the throat, the endpoint analysis of
Sec.~\ref{sec:IR-asymp} applied to numerator and denominator of
\eqref{eq:var-exact} gives
\begin{equation}
\mathcal{V}(\xi)\;\simeq\;\frac{\sqrt{\pi}}{3}\,\xi^{3/2}
\;-\;\sqrt{\pi}\Bigl(\frac{\pi^{2}}{6}-1\Bigr)\xi^{1/2}+\cdots
\;\;\simeq\;\;\frac{1}{6\,\mathcal{G}_0(\xi)}\,,
\qquad \xi\to\infty\,,
\label{eq:var-IR}
\end{equation}
the last form holding because $\mathcal{G}_0'$ and $\mathcal{G}_0''$ are
subleading at the spectral edge. Numerically (Fig.~\ref{fig:variance}; same quadrature as
Appendix~\ref{app:h-eval}): $\mathcal{V}(1)=0.7957$; $\mathcal{V}=1$ at
$\xi_{\rm var}\simeq 1.216$; $\mathcal{V}(\pi)\simeq 3.31$ at the self-dual
point; $\mathcal{V}/\xi^{3/2}\to\sqrt{\pi}/3=0.5908$ (reaching $0.5907$ by
$\xi=10^{4}$), and the two-term form \eqref{eq:var-IR} is accurate to
$3\times10^{-5}$ relative at $\xi=10^{3}$. Kernel fluctuations thus become
comparable to the mean at $\zeta\simeq 1.2\,C$ --- essentially the scale at
which the truncated series of~\cite{Liu:2024} loses control --- and dominate
beyond. In particular, at $\xi\simeq 4.19$, where the truncation places its
spurious confining minimum (Fig.~\ref{fig:G0_compare}), the relative variance
is already $\mathcal{V}\simeq 4.9$: the putative confining wall sits in a
region where the geometry is fluctuation-dominated, compounding the truncation
artifact established in Theorem~\ref{thm:complete-monotone}.

\begin{figure}[t]
  \centering
  \includegraphics[width=0.65\linewidth]{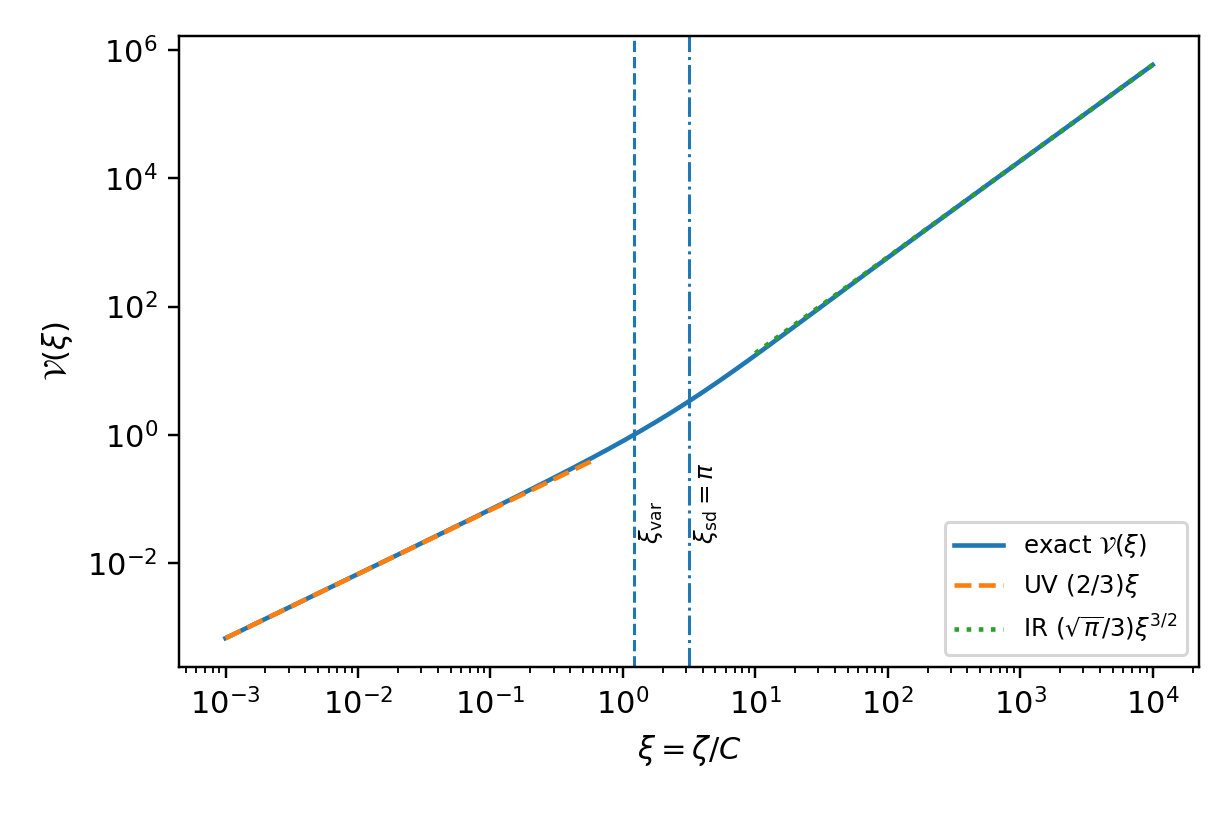}
  \caption{Relative variance $\mathcal{V}(\xi)$ of the Schwarzian kernel.  The
  UV line is the exact leading law $\mathcal{V}\simeq (2/3)\xi$, while the IR
  line is $\mathcal{V}\simeq (\sqrt{\pi}/3)\xi^{3/2}$.  The vertical lines mark
  the point $\xi_{\rm var}\simeq1.216$ where the variance equals the squared
  mean and the self-dual depth $\xi_{\rm sd}=\pi$.}
  \label{fig:variance}
\end{figure}

\medskip\noindent\textit{All moments, spectral-edge universality, and
non-self-averaging.}
The integer moment formula \eqref{eq:hn-def} shows that every moment remains in
the Gaussian-coth, hence Mordell, class.  Deep in the throat all these moments
are controlled by the same spectral edge summarized in \eqref{eq:hDelta-edge}:
although $h_n(\xi)\sim \xi^{2n-3/2}$ after classical normalization,
$\langle\mathcal{B}_1^{\,n}\rangle$ itself falls as $c_n\zeta^{-3/2}$ for
\emph{every} $n$.  Therefore all normalized moment ratios diverge,
\begin{equation}
 \frac{\langle\mathcal{B}^{n}\rangle}{\langle\mathcal{B}\rangle^{n}}
 \sim \xi^{3(n-1)/2},
 \qquad \xi\to\infty .
\label{eq:moment-ratio-IR}
\end{equation}
The distribution of $e^{-\ell_{\rm ren}}$ deep in the throat is heavy-tailed:
its mean is dominated by rare, anomalously short renormalized geodesics, while
by Jensen's inequality~\cite{HardyLittlewoodPolya:1952} $\langle e^{-\ell}\rangle\ge e^{-\langle\ell\rangle}$
the typical throat is \emph{longer} than the annealed geometry suggests.  The
kernel is not self-averaging for $\zeta\gtrsim C$, and
\eqref{eq:quantum-ads2} should be understood there as the mean of a broad
distribution rather than the metric seen by a typical configuration.

\medskip\noindent\textit{What is prescription-robust.}
It is natural to ask which conclusions survive once the averaged-metric
(annealed, $\Delta=1$) choice is relaxed within this family. Repeating the
construction of Sec.~\ref{sec:geometry} with a dimension-$\Delta$ bilocal ---
reading the Weyl factor from $\zeta^{2\Delta}\langle\mathcal{B}_\Delta(2\zeta)
\rangle$ through the same geodesic dictionary, under which a classical
dimension-$\Delta$ two-point function on the dressed metric scales as
$h(\zeta)^{\Delta}/(\Delta t)^{2\Delta}$ --- assigns the metric factor
$h_\Delta(\xi)^{1/\Delta}$. The answer then separates cleanly.

\begin{proposition}[Moment-robust absence of a confining minimum]
\label{prop:moment-robust}
For every $\Delta>0$ let
$h_\Delta(\xi):=\langle\mathcal{B}_\Delta(2\zeta)\rangle/
\mathcal{B}_\Delta^{\rm cl}(2\zeta)$ denote the normalized $\Delta$-moment of
the bilocal, and let $h_\Delta^{1/\Delta}$ be the metric factor a
dimension-$\Delta$ probe assigns. Then
\begin{equation}
\mathcal{G}_0^{(\Delta)}(\xi)\;:=\;\frac{h_\Delta(\xi)^{1/\Delta}}{\xi^{2}}
\label{eq:G0-Delta-def}
\end{equation}
is strictly decreasing on $(0,\infty)$. In particular, no moment-based
effective geometry in this family develops a finite-depth minimum of the
effective string tension.
\end{proposition}

\begin{proof}
$h_\Delta(\xi)^{1/\Delta}/\xi^{2}
=c_\Delta\,\langle\mathcal{B}_\Delta(2\zeta)\rangle^{1/\Delta}$ with
$c_\Delta>0$, and
$\langle\mathcal{B}_\Delta(2\zeta)\rangle
=\int_0^\infty d\omega\,\rho_\Delta(\omega)\,e^{-2\zeta\omega}$ with
$\rho_\Delta(\omega)\propto
\sinh\bigl(2\pi\sqrt{2C\omega}\bigr)\,
\bigl|\Gamma\bigl(\Delta+i\sqrt{2C\omega}\bigr)\bigr|^{4}\ge0$
[the $\Delta$-generalization of \eqref{eq:Gpp-original}] is the Laplace
transform of a nonzero positive measure, hence strictly decreasing in $\zeta$;
composition with the increasing map $x\mapsto x^{1/\Delta}$ preserves strict
decrease.
\end{proof}

The conclusion that the confining minimum of the truncated series is an
artifact is therefore independent of which moment of the fluctuating geometry
one declares to be ``the'' metric. The infrared \emph{exponent}, by contrast,
is not: from \eqref{eq:hDelta-edge},
\begin{equation}
h_\Delta(\xi)^{1/\Delta}\;\sim\;\xi^{\,a(\Delta)},
\qquad a(\Delta)\;=\;2-\frac{3}{2\Delta}\,,
\label{eq:aDelta}
\end{equation}
so the screening law depends on the probe dimension through the
general-exponent analysis of Sec.~\ref{sec:effective-tension}
[cf.\ \eqref{eq:L-exponent-general}]: $\Delta=1$, the kernel of~\cite{Liu:2024}
and of this paper, gives $a=\tfrac12$ and $E\sim-\kappa_{\rm IR}/L^{2}$, while
already at $\Delta=2$ one finds $a=\tfrac54>1$, for which the renormalized
energy integral over the throat (Sec.~\ref{sec:turning-point}) no longer
converges and the structure of the IR analysis changes qualitatively. By
Lyapunov's moment inequality~\cite{HardyLittlewoodPolya:1952}, $h_\Delta^{1/\Delta}$ is nondecreasing in
$\Delta$ at fixed $\xi$: lower moments, which weight typical rather than rare
configurations, are \emph{more} suppressed deep in the throat. Within this
family, the breakdown of mean-field therefore cannot resurrect confinement;
what it relaxes is the specific $L^{-2}$ power, which should be read as an
exact statement about the annealed $\Delta=1$ geometry.

\medskip\noindent\textit{Isotropy, the $2\zeta$ dictionary, and what is
calibrated rather than predicted.}
Two further structural choices enter
\eqref{eq:quantum-ads2}--\eqref{eq:quantum-ads5} and deserve to be explicit.
First, the isotropy of the Weyl factor --- the same $h$ multiplying $dt^{2}$
and $d\zeta^{2}$ --- involves no loss of generality: every static
two-dimensional metric is conformally flat, so isotropy is a choice of frame.
The physical input is the choice of conformal frame itself, made by applying
the \emph{classical} geodesic dictionary --- a boundary bilocal of separation
$\Delta t$ probes maximal Poincar\'e depth $\zeta=\Delta t/2$ --- to the
\emph{quantum} kernel, $h(\zeta)\propto\zeta^{2}G_{\partial\partial}(2\zeta)$.
An $\mathcal{O}(1)$ rescaling of this dictionary, $\Delta t=2\zeta/\lambda$,
acts as $\zeta\to\lambda\zeta$, i.e.\ as a redefinition of $C$. The scaling
exponents central to the rectangular-loop analysis --- the $\xi^{1/2}$ growth
of $h$, the monotone decay of $\mathcal{G}_0$, and the $L^{-2}$ screening law
--- are invariant under this rescaling, but calibrated scales are not: the
crossover $L_c^{\rm match}=82.4685\,C$, the self-dual depth
$\zeta_{\rm sd}=\pi C$, and the coefficient $\kappa_{\rm IR}$ of
\eqref{eq:kappa-IR-restored} all shift, and should be quoted as
prescription-calibrated quantities rather than invariant predictions.

Second,
the transverse metric: as emphasized in Sec.~\ref{sec:effective-tension}, $h$
dresses $g_{tt}$ and $g_{UU}$ but not $g_{xx}$, whose $T^{3}$ volume sits in
the dilaton sector of the reduction. Since the screening exponent follows from
$g_{tt}$ \emph{together with} the $\zeta$-independence of $g_{xx}$, an
independent quantum dressing of the dilaton sector could shift the exponent;
this --- not the isotropy of the AdS$_2$ factor --- is the main structural
assumption beyond mean-field, inherited from the matching ansatz
of~\cite{Liu:2024}. Finally, the bilocal-as-metric construction repackages
averaged boundary-anchored geodesic data as an effective line element;
gauge-invariant formulations of bulk observables and matter Wilson lines in JT
gravity~\cite{Blommaert:2018oro,Blommaert:2019hjr,Mertens:2022review} and the
ensemble treatment of the near-extremal
spectrum~\cite{Iliesiu:2020qvm,Saad:2019lba} provide the natural framework
for going beyond it. The fully averaged observable $\langle W\rangle$ itself
--- the Schwarzian average of the worldsheet functional, with the connected
kernels of \eqref{eq:logW-cumulants} evaluated at split arguments --- is an
interesting problem that we leave for future work.

\medskip\noindent
In summary, the averaged-metric prescription is quantitatively controlled for
$\zeta\lesssim C$, where $\mathcal{V}(\xi)\approx\tfrac{2}{3}\,\xi\ll1$ and the
mean geometry is a faithful proxy for the typical one. Deep in the throat the
kernel is fluctuation-dominated, and the rectangular-loop results of
Sec.~\ref{sec:rectangular} should be read as statements about the annealed mean geometry: their qualitative content ---
monotone tension decay, absence of a confining wall, screening rather than
confinement --- is robust across the moment family by
Proposition~\ref{prop:moment-robust}, while quantitative coefficients and the
precise $L^{-2}$ power are tied to the $\Delta=1$ prescription. The variance
\eqref{eq:var-exact}, being itself exactly Mordell, makes this error budget
computable rather than conjectural.

\section{Rectangular Wilson loop and the quark-antiquark potential}
\label{sec:rectangular}

We now study the temporal rectangular Wilson loop in the quantum-corrected metric~\eqref{eq:quantum-ads5}, within the same averaged-metric probe approximation adopted in~\cite{Liu:2024} and following the general framework of~\cite{Maldacena:1998im,Rey:1998ik,Drukker:1999zq,Kinar:1999,Sonnenschein:2000qm}. We first derive the parametric integrals for the quark separation $L$ and the static energy $E$, perform the standard ultraviolet (UV) subtraction, and then use the exact integral representation of $h(\zeta)$ to determine the large-distance behavior of the quark-antiquark potential.

\subsection{Parametric formulas and renormalization}
\label{sec:renorm-energy}

We now consider a temporal rectangular Wilson loop extended for a time interval $T$ (eventually $T\to\infty$) and with spatial separation $L$ between the quark and antiquark. In the quantum-corrected metric \eqref{eq:quantum-ads5}, choosing worldsheet coordinates $\tau=t$, $\sigma=x$, and an embedding $U=U(x)$, the induced metric leads to the Nambu--Goto action~\cite{Kinar:1999}
\begin{equation}
  S_{\rm NG} =
  \frac{T}{2\pi}\int_{-L/2}^{L/2} dx\,
  \sqrt{h(U)^2(\partial_x U)^2
  + \frac{U^4}{R^4}f(U)h(U)}.
\end{equation}
There is a conserved quantity associated with translations in $x$. Introducing $v\equiv U/U_0$ and the shorthand $f_v(v) = f(U_0 v)$, $h_v(v) = h(U_0 v)$, one finds that the separation $L$ and the static energy $E=S_{\rm NG}/T$ can be written as parametric integrals in terms of $v$:
\begin{align}
  \frac{L}{2} &=
  \frac{R^2}{U_0}
  \int_1^\infty dv\,
  \frac{\sqrt{h_v(v)}}{v^2\sqrt{f_v(v)}}
  \sqrt{\frac{f_v(1)h_v(1)}{v^4 f_v(v)h_v(v)-f_v(1)h_v(1)}},
  \label{eq:L-parametric}\\[4pt]
  E &=
  \frac{U_0}{\pi}\left[
    \int_1^\infty dv
    \left(
      \frac{v^2\,h_v(v)\,\sqrt{f_v(v)h_v(v)}}{\sqrt{v^4 f_v(v)h_v(v)-f_v(1)h_v(1)}}
      - h_v(v)
    \right)
    - \int_{U_T/U_0}^{1} dv\,h_v(v)
  \right].
  \label{eq:E-parametric}
\end{align}
Eliminating $U_0$ between $L$ and $E$ yields the quark-antiquark potential~\cite{Kinar:1999}.

The connected worldsheet action in \eqref{eq:L-parametric}--\eqref{eq:E-parametric} contains the
standard ultraviolet divergence associated with the infinite masses of the external sources~\cite{Maldacena:1998im,Rey:1998ik}.
We regulate the radial integral by cutting off the geometry at $U=U_{\max}$ (equivalently
$\zeta=\epsilon\ll 1$ in the AdS$_2$ coordinate), and define the renormalized potential by subtracting
two straight strings.
Writing the on-shell energy of the connected configuration as an integral over $U$,
\begin{equation}
E_{\rm conn}(U_{\max};U_0)
=
\frac{1}{\pi}\int_{U_0}^{U_{\max}} dU\;
\frac{h(U)}{\sqrt{1-\dfrac{U_0^4 f(U_0)h(U_0)}{U^4 f(U)h(U)}}}\,,
\label{eq:Econn-cutoff}
\end{equation}
where $U_0$ is the turning point.
A straight string at fixed $\vec x$ has induced metric $\gamma_{tt}=g_{tt}$ and $\gamma_{UU}=g_{UU}$, hence
\begin{equation}
\sqrt{\det\gamma}=\sqrt{g_{tt}g_{UU}}=\alpha'\,h(U)\,,
\end{equation}
so the regulated mass of a single external source is
\begin{equation}
m_W(U_{\max})=\frac{1}{2\pi}\int_{U_T}^{U_{\max}} dU\;h(U),
\qquad
2m_W(U_{\max})=\frac{1}{\pi}\int_{U_T}^{U_{\max}} dU\;h(U).
\label{eq:mW-cutoff}
\end{equation}
We therefore define the renormalized potential by
\begin{equation}
E(U_0)\equiv
\lim_{U_{\max}\to\infty}\Big(E_{\rm conn}(U_{\max};U_0)-2m_W(U_{\max})\Big).
\label{eq:Eren-def}
\end{equation}
Using \eqref{eq:Econn-cutoff}--\eqref{eq:mW-cutoff} one can rearrange the difference \emph{exactly} as
\begin{equation}
E(U_0)
=
\frac{1}{\pi}\int_{U_0}^{\infty} dU\;
\bigg[
\frac{h(U)}{\sqrt{1-\dfrac{U_0^4 f(U_0)h(U_0)}{U^4 f(U)h(U)}}}-h(U)
\bigg]
\;-\;
\frac{1}{\pi}\int_{U_T}^{U_0} dU\;h(U).
\label{eq:Eren-finite}
\end{equation}
The first integral is UV-finite: at large $U$ one has $f(U)\to 1$ and $h(U)\to 1$, and the bracket admits the expansion
\begin{equation}
\frac{h(U)}{\sqrt{1-\dfrac{U_0^4 f(U_0)h(U_0)}{U^4 f(U)h(U)}}}-h(U)
=
\frac{U_0^4 f(U_0)h(U_0)}{2U^4}+O(U^{-6}),
\label{eq:UV-decay}
\end{equation}
which is integrable at $U=\infty$.

To display the cancellation in the $\zeta$-cutoff language, use the near-boundary series
$h(\zeta)=1+O((\zeta/C)^2)$ from Eq.~\eqref{eq:h-small-zeta} and the matching relation
$U-U_T=R^2/(12\zeta)$ (Eq.~\eqref{eq:zetaUmap}), so that $dU=-(R^2/12)\,d\zeta/\zeta^2$.
The UV cutoff $U_{\max}$ maps to a small $\zeta$-cutoff $\epsilon>0$, and the integral from some
finite $U$ to the cutoff becomes
\begin{equation}
\int^{U_{\max}} dU\,h(U)
\;=\;
\frac{R^2}{12}\int_{\epsilon}^{\zeta_{\rm fin}} \frac{d\zeta}{\zeta^2}\;\bigl[1+O\bigl((\zeta/C)^2\bigr)\bigr]
\;=\;
\frac{R^2}{12}\bigg(\frac{1}{\epsilon}+O(1)\bigg),
\label{eq:UV-pole}
\end{equation}
where $\zeta_{\rm fin}$ is a fixed finite value (set by $U_T$ or $U_0$) that does not contribute to
the divergence. Since both $E_{\rm conn}(U_{\max};U_0)$ and $2m_W(U_{\max})$ contain the same $1/\epsilon$ pole,
\eqref{eq:Eren-def} is finite.
The renormalized potential~\eqref{eq:Eren-finite} receives contributions of competing sign: the first integral is positive while the subtraction term is negative. Both contribute at the same $\zeta_0^{-1/2}$ order in the deep IR; the full numerical
coefficient of the leading deep-IR behavior is therefore determined only after combining both
terms in the throat analysis (Sec.~\ref{sec:turning-point}). Combining the two throat contributions gives a negative overall coefficient (confirmed numerically in Sec.~\ref{subsec:numerical-branch-selection}, where $E<0$ throughout the sampled branch); the subtraction term provides the simplest way to exhibit the scaling:
\begin{equation}
E(U_0)\;\sim\;
-\frac{1}{\pi}\int_{U_T}^{U_0} dU\;h(U)\;<\;0.
\end{equation}
Using the asymptotic~\eqref{eq:h-IR-from-coth} together with
$U-U_T\propto \zeta^{-1}$ in the AdS$_2$ throat (Eq.~\eqref{eq:zetaUmap}), one has
\begin{equation}
h(U)\;\sim\;A\,(U-U_T)^{-1/2}\qquad(U\to U_T^+),
\end{equation}
and therefore
\begin{equation}
E(U_0)\;\sim\;-\frac{A}{\pi}\int_{0}^{U_0-U_T}\!dw\;w^{-1/2}
\;=\;-\frac{2A}{\pi}\,\sqrt{U_0-U_T}
\;\propto\;
-\zeta_0^{-1/2},
\label{eq:Eren-IR-scaling}
\end{equation}
where $\zeta_0$ is the turning point in AdS$_2$ coordinates.
A detailed turning-point derivation, establishing the relation $L\propto \zeta_0^{1/4}$ and hence
$E(L)\propto -L^{-2}$, is given in Sec.~\ref{sec:turning-point}.

\subsection{Short distances: Coulomb, analytic series, and nonperturbative corrections}
\label{sec:short-distance}

At short distances the turning point $U_0$ lies far from the horizon, $h\to 1$,
and the standard conformal Coulomb result~\cite{Maldacena:1998im,Rey:1998ik} is
recovered:
\begin{equation}
  E(L)\;\underset{L\to 0}{\sim}\;
  -\kappa\,\frac{R^2}{L}\,,
  \qquad
  \kappa \;=\;
  \frac{4\pi^2}{\Gamma(1/4)^4}\;\simeq\; 0.2285\,.
  \label{eq:Coulomb-coeff}
\end{equation}
As $L$ grows and the turning point enters the near-horizon region, the quantum
corrections encoded in $h(\zeta)$ deform this Coulomb tail.  The exact integral
representation implies the following structure for the small-$L$ expansion of the potential, generalizing the analysis of~\cite{Liu:2024}:
\begin{equation}
  E(L) = -\kappa\frac{R^2}{L}
  + \sigma_{\rm rect}\,\frac{R^2}{C^2}L
  + \gamma_3\,\frac{R^2}{C^4}L^3
  + \gamma_5\,\frac{R^2}{C^6}L^5
  + \cdots
  + \delta_{\rm mod}(L).
  \label{eq:E-smallL}
\end{equation}
The polynomial part is determined entirely by the near-boundary
expansion~\eqref{eq:h-exact-Taylor}.  In the smooth matched ansatz
\eqref{eq:smooth-matched-ansatz}, the leading short-distance deformation is
fixed by
\begin{equation}
        h(U)=1+\frac{\beta}{U^2}+O(U^{-3}),
        \qquad
        \beta=\frac{c_2}{144}=\frac{1}{8640},
\label{eq:beta-short-distance}
\end{equation}
where $\xi=1/[12(U-U_T)]$ and $U_T=R=C=1$.  Expanding the parametric
integrals to first order in $\beta/U_0^2$ gives a matched-ansatz calibration of
the linear coefficient.  With
\begin{align}
 I_0&:=\int_1^\infty\frac{dv}{v^2\sqrt{v^4-1}},
        \qquad \ell_0:=2I_0,
\label{eq:I0-short-distance}\\
 I_1&:=\frac12\int_1^\infty
 \frac{dv}{v^2\sqrt{v^4-1}}
 \left(1+\frac{1}{v^2}-\frac{1}{v^2+1}\right),
\label{eq:I1-short-distance}
\end{align}
one finds
\begin{equation}
 \sigma_{\rm rect}^{\rm match}
 =\frac{c_2}{144}
 \left(\frac{2I_1}{\pi\ell_0}+\frac{1}{4\pi}\right)
 =3.3593845\times10^{-5}.
\label{eq:sigma-rect-match}
\end{equation}
This is the term that, in finite-order analyses such as~\cite{Liu:2024}, is
identified with a string tension reminiscent of the Cornell potential
\cite{Eichten:1975,Eichten:1978,Bali:2000gf}.  Its value is not universal: it
is a calibration of the chosen matching prescription, whereas the $L^{-2}$ tail
below is fixed by the exact throat.

Our refined expansion~\eqref{eq:E-smallL} is consistent with the short-distance
matching, but the linear term is not fundamental.  Higher analytic corrections
modify the Cornell form~\cite{Eichten:1975,Eichten:1978} outside the strict UV regime, and at large distances the
potential crosses over to the screened regime derived in the next subsection.

The correction $\delta_{\rm mod}(L)$ represents the effect of exponentially
suppressed terms of order $e^{-\pi^2 C/\zeta}$ in $h(\zeta)$, which arise from
the $S$-transformed ($q^\ast$) channel of the Mordell evaluation
(Sec.~\ref{sec:mordell-eval}, Eq.~\eqref{eq:h-q-series}) and are invisible to any finite Taylor truncation
in $\zeta/C$.  Their structure in terms of the Appell--Lerch sums $L_r(q^\ast)$ is
made explicit by the $q^\ast$-dependent terms in~\eqref{eq:h-q-series}, but the precise functional form of
$\delta_{\rm mod}$ as a function of $L$ depends on the mapping between the
AdS$_2$ coordinate $\zeta$ and the boundary separation $L$, which in turn relies
on the matching of the near-horizon and far metrics~\cite{Liu:2024}; its
determination is left for future work.
The exponential scale $e^{-\pi^2 C/\zeta}$ is connected to the optimal
truncation of the asymptotic near-boundary series (Sec.~\ref{sec:exact-Taylor}):
at order $m_{\rm opt}\sim\pi^2/\xi$ the truncation error is minimized, and the
remaining irreducible error is precisely of the modular/nonperturbative order
$q^\ast=e^{-\pi^2/\xi}$.

\subsection{Long distances: algebraic screening}
\label{sec:screening}

From the deep-interior asymptotic~\eqref{eq:h-IR-from-coth}, the AdS$_2$ warp factor behaves for
$\zeta\gg C$ as
\begin{equation}
  a(\zeta)\;:=\;\frac{L_2^2}{\zeta^2}\,h(\zeta)\;\sim\;\frac{A}{\zeta^{3/2}},\qquad A>0\,.
\end{equation}
This is the only input needed to determine the long-distance Wilson-loop tail.
A detailed turning-point analysis (Sec.~\ref{sec:turning-point}) shows that
$L\propto\zeta_0^{1/4}$ and $E\propto -\zeta_0^{-1/2}$, giving the
algebraically screened regime
\begin{equation}
  E(L)\;\sim\; -\frac{\kappa_{\rm IR}}{L^2},\qquad F(L)= -\frac{dE}{dL}\sim -\frac{2\kappa_{\rm IR}}{L^3},\qquad L\gg C.
  \label{eq:E_IR}
\end{equation}
In particular, the potential vanishes as $L\to\infty$ and approaches zero from below. Thus the deep IR does
not enforce a transition to disconnected worldsheets; any Gross--Ooguri-type transition~\cite{Gross:1998} (if present) would
have to arise from intermediate scales.  The numerical check of
Sec.~\ref{subsec:numerical-branch-selection} finds no such pre-emption in the
minimal smooth matched ansatz.
At any nonzero temperature one expects the familiar screening
transition to reappear.

The physical interpretation of \eqref{eq:E_IR} is therefore sharper than the
exponent alone.  The extremal RN brane is governed in the IR by the
AdS$_2\times\mathbb{R}^3$ semi-local quantum liquid: the time direction is
critical, while the spatial directions are spectators of the $z\to\infty$
scaling~\cite{Faulkner:2009wj,Iqbal:2011ae}.  A thermal plasma with a finite
screening mass would give a Debye tail, schematically $E_{\rm conn}(L)\sim
\exp[-m_D(T)L]$, or a connected branch pre-empted by a disconnected saddle.  By
contrast, the infinitely long extremal throat supplies no thermal cap, and the
connected string samples a scale-free IR geometry; the result is the power-law
screening $E(L)\propto -L^{-2}$.  In this sense the algebraic tail is a sharp
Wilson-loop diagnostic of the semi-local critical IR, while the expected
exponential screening at any $T>0$ is the corresponding diagnostic that the
critical throat has been cut off.

\medskip\noindent\textit{Absence of the mass gap.}
The $L^{-2}$ screening forces a reinterpretation of the ``mass gap'' reported in the
truncated analysis of~\cite{Liu:2024}. In the truncated perturbative treatment, the
confinement indicator $\mathcal{G}_0(\zeta):=h(\zeta)/\zeta^2$ (following~\cite{Liu:2024}) develops a minimum at finite
radial depth (Fig.~\ref{fig:G0_compare}), and the curvature of the effective potential at
this minimum defines a characteristic mass scale $m_{\rm gap}$ for radial string
fluctuations. In the exact throat geometry, this minimum is absent: $\mathcal{G}_0$ decays
monotonically as $\zeta^{-3/2}$ and the effective string tension
$\mathcal{T}(\zeta)\propto \zeta^{-3/4}$ decreases monotonically toward zero as the string descends into the AdS$_2$
throat (see Sec.~\ref{sec:effective-tension} below). There is consequently no stable minimum around which
bound-state fluctuations could be quantized, and no evidence for a discrete gapped tower of radial string excitations in the exact throat analysis. The scale previously identified as $m_{\rm gap}$ retains a weaker meaning: it corresponds to
the curvature scale of the effective potential at the crossover region where the approximately linear plateau gives way
to the screened regime. The monotone decay of the effective string tension reflects the gapless nature of the extremal
horizon~\cite{Maldacena:2016upp}, strongly modified by the quantum corrections encoded
in the scaling~\eqref{eq:h-IR-from-coth} deep in the throat.

\medskip\noindent\textit{Analytic proof of monotone decay.}

\begin{theorem}[Complete monotonicity of the confinement indicator]
\label{thm:complete-monotone}
The confinement indicator $\mathcal{G}_0(\xi):=h(\xi)/\xi^2$ is completely monotone
on $(0,\infty)$ in the sense of Bernstein~\cite{Widder:1941}: for every integer $n\ge 0$,
\begin{equation}
(-1)^n\mathcal{G}_0^{(n)}(\xi)
=2\int_0^\infty dy\,y^{3+2n}\,e^{-\xi y^2}\,\coth(\pi y)>0\,,
\qquad \xi>0.
\label{eq:G0-complete-monotone}
\end{equation}
In particular, $\mathcal{G}_0$ is strictly positive and strictly decreasing, and
admits no finite-$\xi$ minimum.
\end{theorem}

\begin{proof}
From the exact integral representation~\eqref{eq:h-exact-normalized},
\begin{equation}
\mathcal{G}_0(\xi)=\frac{h(\xi)}{\xi^2}
=2\int_0^\infty dy\,y^3\,e^{-\xi y^2}\,\coth(\pi y)\,.
\end{equation}
The integrand $y^3\,e^{-\xi y^2}\coth(\pi y)$ is strictly positive for all $y>0$
and $\xi>0$, and the integral converges absolutely, so $\mathcal{G}_0(\xi)>0$.
Differentiation under the integral sign (justified by the Gaussian decay) gives
\begin{equation}
\mathcal{G}_0'(\xi)
=-2\int_0^\infty dy\,y^5\,e^{-\xi y^2}\,\coth(\pi y)<0
\qquad\text{for all }\xi>0\,,
\label{eq:G0-monotone}
\end{equation}
and more generally, the $n$-th derivative brings down $(-y^2)^n$, yielding
\eqref{eq:G0-complete-monotone} with a manifestly positive integrand.
\end{proof}

\begin{remark}
In the notation of Kinar, Schreiber, and Sonnenschein~\cite{Kinar:1999}, the
temporal effective string tension satisfies
$f_{\rm string}^2=g_{tt}g_{xx}\propto h(\zeta)/\zeta^2=\mathcal{G}_0(\zeta)$.
Using $s=U-U_T\sim\zeta^{-1}$ one has
$f_{\rm string}(s)\sim s^{3/4}$, so $f_{\rm string}(0)=0$.
Thus the exact throat belongs to the Kinar--Schreiber--Sonnenschein nonconfining power-law branch, not the
$f(0)>0$ confining branch.  Complete monotonicity of $\mathcal{G}_0$ means that
no finite-depth positive minimum of $f_{\rm string}$ can arise, ruling out
confinement at the level of the exact throat geometry.
\end{remark}

\medskip\noindent\textit{Prescription robustness.}
Theorem~\ref{thm:complete-monotone} is a statement about the mean
($\Delta=1$, annealed) geometry. Proposition~\ref{prop:moment-robust} of
Sec.~\ref{sec:variance} extends the decisive part --- strict monotone decay of
the confinement indicator, hence the absence of any finite-depth minimum --- to
the full family of moment-based effective geometries $h_\Delta^{1/\Delta}$, so
the identification of the confining minimum as a truncation artifact does not
rest on the averaged-metric choice. The quantitative reach of the mean geometry
itself is delimited by the exact variance \eqref{eq:var-exact}: relative
kernel fluctuations reach unity at $\xi\simeq1.22$ and grow as
$(\sqrt{\pi}/3)\,\xi^{3/2}$ beyond, so the $-\kappa_{\rm IR}/L^{2}$ law is the
exact screening law of the annealed geometry, with the prescription dependence
of the exponent quantified in Sec.~\ref{sec:variance}.

\subsection{Detailed derivation of the IR scaling}
\label{sec:turning-point}

The deep-infrared tail $E(L)\sim -\kappa_{\rm IR}/L^2$ follows from a systematic
turning-point expansion of the Wilson-loop parametrics~\cite{Kinar:1999,Sonnenschein:2000qm}. We present this derivation
in two complementary forms: (i) a pure $\zeta$-coordinate analysis in the AdS$_2$ throat,
and (ii) an effective-tension formulation that isolates the geometric origin of the exponents.

\medskip\noindent\textit{Pure $\zeta$-coordinate derivation.}

In the near-horizon region, the geometry is AdS$_2\times T^3$ with quantum-corrected
AdS$_2$ factor (cf.\ Sec.~\ref{sec:geometry}):
\begin{equation}
ds^2_{\rm throat}\approx
h(\zeta)\,\frac{L_2^2}{\zeta^2}(dt^2+d\zeta^2)
+\ell_x^2\,d\vec x^{\,2},
\qquad
\ell_x^2:=\frac{u_T^2}{L_{\rm AdS}^2}.
\label{eq:throat-metric}
\end{equation}
We compute the rectangular Wilson loop using worldsheet coordinates
$\tau=t$, $\sigma=x$ and embedding $\zeta=\zeta(x)$.

The induced metric gives the Euclidean Nambu--Goto density
$\mathcal{L}=\sqrt{g_{tt}}\sqrt{g_{xx}+g_{\zeta\zeta}\zeta'^2}$~\cite{Kinar:1999}.
Because $\mathcal{L}$ has no explicit $x$-dependence, the conserved first integral gives
\begin{equation}
\frac{dx}{d\zeta}
=\frac{L_2}{\ell_x}\frac{\sqrt{h(\zeta)}}{\zeta}\;
\frac{1}{\sqrt{\dfrac{h(\zeta)\zeta_0^2}{h(\zeta_0)\zeta^2}-1}}.
\label{eq:dx-dzeta-explicit}
\end{equation}

In the deep interior, $h(\zeta)\simeq a\,\zeta^{1/2}$ with $a>0$, so
$\sqrt{h(\zeta)}\simeq \sqrt{a}\,\zeta^{1/4}$ and
\begin{equation}
\frac{h(\zeta)\zeta_0^2}{h(\zeta_0)\zeta^2}
\simeq
\frac{a\zeta^{1/2}\,\zeta_0^2}{a\zeta_0^{1/2}\,\zeta^2}
=
\left(\frac{\zeta_0}{\zeta}\right)^{3/2}.
\end{equation}
Plugging into \eqref{eq:dx-dzeta-explicit} and rescaling $\zeta=\zeta_0 u$,
\begin{align}
\frac{L}{2}
&\simeq
\frac{L_2\sqrt{a}}{\ell_x}
\int_0^{\zeta_0}d\zeta\;
\frac{\zeta^{-3/4}}{\sqrt{\left(\frac{\zeta_0}{\zeta}\right)^{3/2}-1}}
\notag\\
&=
\frac{L_2\sqrt{a}}{\ell_x}\;
\zeta_0^{1/4}
\int_0^{1}du\;
\frac{1}{\sqrt{1-u^{3/2}}}.
\end{align}
The remaining integral is a Beta function:
\begin{equation}
\int_0^{1}\frac{du}{\sqrt{1-u^{3/2}}}
=\frac{2}{3}\,B\!\left(\frac{2}{3},\frac{1}{2}\right)
=4\sqrt{\pi}\,\frac{\Gamma(2/3)}{\Gamma(1/6)}.
\end{equation}
Hence the scaling is explicit:
\begin{equation}
L \propto \zeta_0^{1/4}.
\label{eq:L-zeta-scaling}
\end{equation}
The force on the quark-antiquark pair is $F(L)=-dE/dL$, where
\begin{equation}
-F\!\big(L\big)
=\frac{1}{2\pi}\,\sqrt{g_{tt}g_{xx}}\Big|_{\zeta_0}
=\frac{1}{2\pi}\,\frac{L_2}{\zeta_0}\,\sqrt{\ell_x^2\,h(\zeta_0)}.
\end{equation}
With $h(\zeta_0)\sim a\zeta_0^{1/2}$ this gives
$|F|\propto \zeta_0^{-3/4}$.
Using $L\propto \zeta_0^{1/4}\Rightarrow \zeta_0\propto L^4$, we obtain
\begin{equation}
F(L)\propto -L^{-3}
\qquad\Rightarrow\qquad
E(L)\propto -L^{-2},
\end{equation}
consistent with the algebraic screening regime.

\medskip\noindent\textit{Effective-tension viewpoint.}
\label{sec:effective-tension}

The scaling can be understood directly by introducing the
\emph{local effective string tension}~\cite{Kinar:1999}
\begin{equation}
\mathcal{T}(\zeta)\ :=\ \frac{1}{2\pi\alpha'}\sqrt{g_{tt}(\zeta)\,g_{xx}}
\ \propto\ \sqrt{g_{tt}(\zeta)\,g_{xx}}.
\label{eq:T-eff-def}
\end{equation}
Since $g_{xx}=\ell_x^2$ is constant in the throat and $g_{tt}=L_2^2\,h(\zeta)/\zeta^2$, we have
\begin{equation}
\mathcal{T}(\zeta)\ \propto\ \frac{\sqrt{h(\zeta)}}{\zeta}.
\label{eq:T-eff-h}
\end{equation}
The conserved quantity from $x$-translations can be recast as
\begin{equation}
\frac{dx}{d\zeta}
=
\sqrt{\frac{g_{\zeta\zeta}}{g_{xx}}}\;
\frac{1}{\sqrt{\left(\dfrac{\mathcal{T}(\zeta)}{\mathcal{T}(\zeta_0)}\right)^2-1}},
\qquad
\frac{L}{2}=\int_0^{\zeta_0}d\zeta\;\frac{dx}{d\zeta}.
\label{eq:L-from-T-eff}
\end{equation}
This is the ``effective potential'' form: $\zeta'(x)^2$ is controlled by the ratio
$\mathcal{T}(\zeta)/\mathcal{T}(\zeta_0)$, and the turning point occurs where the bracket
vanishes.

Because of the AdS$_2$ throat (double zero of $f$), we have the universal relation
$g_{tt}=g_{\zeta\zeta}\propto h(\zeta)/\zeta^2$.
For the deep-IR asymptotic $h(\zeta)\sim \zeta^{a}$, this gives
\begin{equation}
g_{tt}(\zeta)\ \propto\ \frac{h(\zeta)}{\zeta^2}\ \sim\ \zeta^{a-2}
\equiv \zeta^{-p},
\qquad p:=2-a,
\label{eq:gtt-power}
\end{equation}
and from \eqref{eq:T-eff-h},
$\mathcal{T}(\zeta)\propto \zeta^{-p/2}$.
Evaluating \eqref{eq:L-from-T-eff} with $g_{\zeta\zeta}=g_{tt}\sim \zeta^{-p}$ and $g_{xx}=\ell_x^2$,
and rescaling $\zeta=\zeta_0 u$, one finds
\begin{equation}
L \propto \zeta_0^{\,1-p/2}=\zeta_0^{\,a/2}.
\label{eq:L-exponent-general}
\end{equation}
Thus \emph{the exponent in $L(\zeta_0)$ is fixed entirely by the single IR power of}
$h(\zeta)$ \emph{together with AdS$_2$ scaling} (which enforces the $\zeta^{-2}$ factor in 
$g_{tt}$).

In the present problem, $a=\tfrac12$ and $p=\tfrac32$, giving immediately
\begin{equation}
L\ \propto\ \zeta_0^{1/4},
\qquad
\mathcal{T}(\zeta)\ \propto\ \zeta^{-3/4}.
\label{eq:L-T-Mordell}
\end{equation}
The effective tension decreases monotonically as the string descends deeper into the throat,
with no minimum---this is the geometric origin of the absence of a confining regime.
A common AdS$_2$ intuition~\cite{Maldacena:1998im} is that for a metric of the form
$ds^2\sim (dx^2+d\zeta^2)/\zeta^2$ one finds a turning-point relation $L\propto \zeta_0$
(the geodesic ``semicircle'' picture), which would suggest that an estimate of the form
$E\sim L\,\mathcal{G}_0(\zeta_0)$ with $\mathcal{G}_0=h(\zeta)/\zeta^2$ leads to a fractional
power law $E(L)\sim L^{-1/2}$ when $h(\zeta)\sim(\zeta/C)^{1/2}$.
However, this reasoning does \emph{not} apply in the present background: in the AdS$_2\times T^3$
throat \eqref{eq:throat-metric} the separation directions lie in the transverse $T^3$ factor, so
$g_{xx}$ is $\zeta$-independent, and in the full string-frame metric \eqref{eq:quantum-ads5} the
quantum factor $h$ multiplies $g_{tt}$ and $g_{UU}$ but \emph{not} $g_{xx}$.

\subsection{Numerical check of the connected branch in a smooth matched ansatz}
\label{subsec:numerical-branch-selection}

The preceding derivation establishes the local deep-throat saddle.  A separate
question, familiar from holographic Wilson-loop analyses, is whether this saddle
lies on the physical branch of the global parametric curve.  In particular,
multi-branch structure can occur when the string probes different radial regimes,
and the physical branch must satisfy the usual monotonicity and concavity
conditions
\begin{equation}
        \frac{dL}{dU_0}<0,\qquad
        \frac{dE}{dL}>0,\qquad
        \frac{d^2E}{dL^2}\leq 0,
\label{eq:branch-concavity-conditions}
\end{equation}
where $U_0$ is the turning point in a radial coordinate increasing toward the
asymptotic AdS$_5$ boundary. (The condition $dE/dL>0$ is consistent with the
attractive force $F(L)=-dE/dL<0$; note that $E(L)<0$ and approaches zero from
below, so $E$ is an increasing function of $L$ in the screened regime.)  Equivalently, in the notation of the stability
analyses of classical string embeddings, a branch with the opposite sign of
$dL/dU_0$ is a natural candidate for a fluctuation instability
\cite{Arias:2009yz,Chatzis:2024xrn}.  Examples with several competing branches,
maximal separations, or first-order branch changes have been studied in related
D3-brane geometries \cite{Brandhuber:1998bs}; recent Wilson-loop computations
also emphasize the usefulness of global numerical minimization when conformal,
confining, and screened regimes coexist \cite{Giliberti:2024dpm}.

We therefore performed a direct numerical check (full implementation details, tolerances, and
cutoff-dependence tests are collected in Appendix~\ref{app:numerics}) using the exact integral
representation of $h$.  First, in the pure AdS$_2$ throat, the result can be packaged as a
universal normalized screening curve.  In units $C=L_2=\ell_x=1$, with the
boundary condition $h(0)=1$ fixed as in \eqref{eq:h-exact-normalized}, the
separation and renormalized energy are
\begin{align}
L_{\rm th}(\xi_0)
 &=2\int_0^1du\,
   \left(\frac{h_0\,h(\xi_0u)}{h(\xi_0u)-h_0u^2}\right)^{1/2},
\label{eq:numerical-L-throat}
\\[2mm]
E_{\rm th}(\xi_0)
 &= \frac{1}{\pi}\left[
 \frac{1}{\xi_0}\int_0^1du\,\frac{h(\xi_0u)}{u^2}
 \left(\frac{1}{\sqrt{1-h_0u^2/h(\xi_0u)}}-1\right)
 -\int_{\xi_0}^{\infty}d\xi\,\frac{h(\xi)}{\xi^2}
 \right],
\label{eq:numerical-E-throat}
\end{align}
where $h_0:=h(\xi_0)$.  The endpoint singularity in
\eqref{eq:numerical-E-throat} is integrable and was treated by a change of
variables in the numerical code.  Equations~\eqref{eq:numerical-L-throat}--\eqref{eq:numerical-E-throat}
should be read as more than a numerical check: they define a reusable universal
function, independent of the UV completion,
\begin{equation}
        \mathfrak{E}_{\rm th}(\mathfrak{L})
        :=E_{\rm th}(\xi_0(\mathfrak{L})),
        \qquad
        \mathfrak{L}:=L_{\rm th}(\xi_0),
        \qquad \xi_0\in(0,\infty).
\label{eq:universal-throat-function}
\end{equation}
This curve has an exact short-endpoint structure.  If the throat is classical,
$h\equiv1$, then
\begin{equation}
        L_{\rm th}=2\int_0^1\frac{du}{\sqrt{1-u^2}}=\pi,
        \qquad E_{\rm th}=0,
\label{eq:classical-throat-marginal}
\end{equation}
for every turning point: the connected and disconnected saddles are marginally
degenerate at the single separation $\mathfrak{L}=\pi$.  The full Schwarzian
curve is therefore quantum-generated.  As $\xi_0\to0$, one finds
\begin{equation}
        \mathfrak{E}_{\rm th}(\pi)
        =-\frac{1}{\pi}\int_0^\infty d\xi\,\frac{h(\xi)-1}{\xi^2}
        =-\frac{1}{6\pi},
\label{eq:throat-threshold-energy}
\end{equation}
where the last equality follows directly from
\begin{equation}
 \int_0^\infty d\xi\,\frac{h(\xi)-1}{\xi^2}
 =4\int_0^\infty \frac{y\,dy}{e^{2\pi y}-1}
 =\frac{4\zeta_{\rm R}(2)}{(2\pi)^2}
 =\frac{1}{6}.
\label{eq:threshold-one-sixth}
\end{equation}
Numerically $\mathfrak{L}-\pi=O(\xi_0^2)$ while
$\mathfrak{E}_{\rm th}+1/(6\pi)=O(\xi_0)$, so the pure-throat curve leaves the
threshold with a square-root branch point.  In the isolated throat
$\mathfrak{L}=\pi$ is a minimal separation; in a matched AdS$_5$/AdS$_2$
geometry, smaller separations are supplied by the far-region branch.

Restoring units for any geometry whose connected string enters the same
Schwarzian throat gives
\begin{equation}
        L=\frac{L_2}{\ell_x}\,\mathfrak{L},
        \qquad
        E=\frac{L_2^2}{C}\,\mathfrak{E}_{\rm th}(\mathfrak{L}),
\label{eq:universal-throat-scales}
\end{equation}
with all model dependence confined to the matching onto this curve and to the
identification of $C$, $L_2$, and $\ell_x$.  At the threshold,
$E_{\rm th}=-L_2^2/(6\pi C)$; this binding energy vanishes in the classical
limit $C\to\infty$.  The large-$\mathfrak{L}$ endpoint of the universal function
is the semi-local screening law
$\mathfrak{E}_{\rm th}(\mathfrak{L})\sim -\kappa_{\rm th}/\mathfrak{L}^2$.
These formulae give
\begin{equation}
        (-E_{\rm th})L_{\rm th}^{2}=0.26003\quad(\xi_0>10^3),
        \qquad
        \frac{d\log(-E_{\rm th})}{d\log L_{\rm th}}=-2.001,
\label{eq:throat-numerical-fit}
\end{equation}
which agrees with the analytic coefficient
\begin{equation}
        \kappa_{\rm th}= \frac{32}{\sqrt{\pi}}
        \left(\frac{\Gamma(2/3)}{\Gamma(1/6)}\right)^3
        =0.259920745\ldots .
\label{eq:kappa-throat-numerical}
\end{equation}
Figure~\ref{fig:universal-throat} shows the full universal curve and both exact
endpoints.  In particular, $(-\mathfrak{E}_{\rm th})\mathfrak{L}^2$
runs from the threshold value $\pi/6$ to $\kappa_{\rm th}$ as the turning
point moves from the boundary to the deep throat.

\begin{figure}[t]
  \centering
  \includegraphics[width=0.65\linewidth]{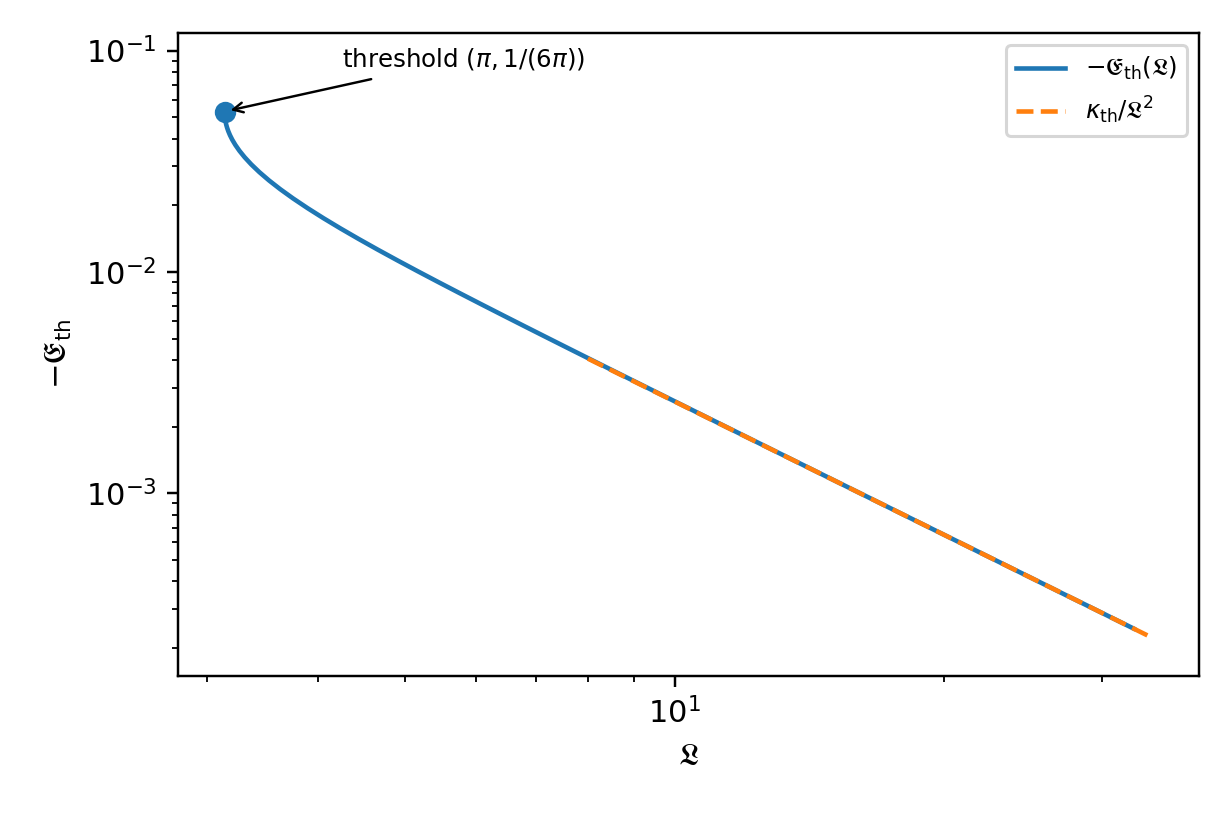}
  \caption{Universal throat screening curve in the normalized units
  $C=L_2=\ell_x=1$.  The plotted quantity is
  $-\mathfrak{E}_{\rm th}(\mathfrak{L})$.  The left endpoint is the exact
  quantum threshold $(\mathfrak{L},-\mathfrak{E}_{\rm th})=(\pi,1/(6\pi))$;
  the dashed line is the large-separation tail
  $\kappa_{\rm th}/\mathfrak{L}^2$.}
  \label{fig:universal-throat}
\end{figure}

Equivalently, the coefficient in physical units is
\begin{equation}
        \kappa_{\rm IR}=
        \frac{32}{\sqrt{\pi}}
        \left(\frac{\Gamma(2/3)}{\Gamma(1/6)}\right)^3
        \frac{L_2^4}{\ell_x^2 C}\,.
\label{eq:kappa-IR-restored}
\end{equation}
Here $L_2$ is the AdS$_2$ radius (length), $\ell_x=u_T/L_{\rm AdS}$ is the
dimensionless horizon scale entering the transverse metric, and $C$ is the JT
heat capacity (length), so $\kappa_{\rm IR}$ has dimensions of $[{\rm
length}]$, consistent with $E(L)\sim -\kappa_{\rm IR}/L^2$ for a renormalized
energy $E$ of dimension $[{\rm length}]^{-1}$.  (The normalized curve of
Sec.~\ref{sec:renorm-energy} is obtained in units $C=L_2=\ell_x=1$, in which
$E$ and $L$ are measured in the same length unit.)

Second, to check that the connected saddle remains dominant across the AdS$_5$/AdS$_2$ interpolation, we
used the minimal smooth completion
\begin{equation}
        f(U)=1-\frac{3}{U^4}+\frac{2}{U^6},
        \qquad
        \xi(U)=\frac{1}{12C(U-U_T)},
        \qquad
        h(U)=h_{\rm exact}(\xi(U)),
\label{eq:smooth-matched-ansatz}
\end{equation}
with $U_T=R=C=1$.  This ansatz is not a substitute for a first-principles
solution of the quantum-corrected interpolation; it is the simplest smooth
geometry that has the exact Schwarzian throat in the IR and tends to $h=1$ in
the UV.  The parametric integrals of Sec.~\ref{sec:renorm-energy} were then
evaluated over
\begin{equation}
        U_0-U_T\in[10^{-7},10^2]\,.
\end{equation}
Over this range no cusp or maximal separation was found, and the sign conditions~\eqref{eq:numerical-sign-checks} hold (Fig.~\ref{fig:numerical-branch-selection}):
\begin{equation}
        \frac{dL}{dU_0}<0,
        \qquad
        \frac{dE}{dL}>0,
        \qquad
        \frac{d^2E}{dL^2}<0
\label{eq:numerical-sign-checks}
\end{equation}
at every sampled point.  Since the near-horizon radius of the extremal
RN-AdS$_5$ throat in these units is $L_2^2=1/12$, Eq.~\eqref{eq:kappa-IR-restored}
predicts
\begin{equation}
        \kappa_{\rm RN}=\frac{\kappa_{\rm th}}{144}
        =1.805005\times 10^{-3}\,.
\label{eq:kappa-rn-prediction}
\end{equation}
The numerical branch gives
\begin{equation}
        \langle (-E)L^2\rangle_{\xi_0>10^4}=1.80408\times 10^{-3},
        \qquad
        \frac{d\log(-E)}{d\log L}=-2.003,
\label{eq:full-branch-numerical-fit}
\end{equation}
confirming that, in this matched ansatz, the connected branch remains monotone,
concave, and reaches the algebraically screened IR regime.  With the subtraction
convention in which the disconnected pair has zero energy, the numerical energy
is negative throughout the sampled branch, so no disconnected saddle pre-empts it
within this ansatz.

\begin{figure}[t]
    \centering
    \includegraphics[width=0.95\textwidth]{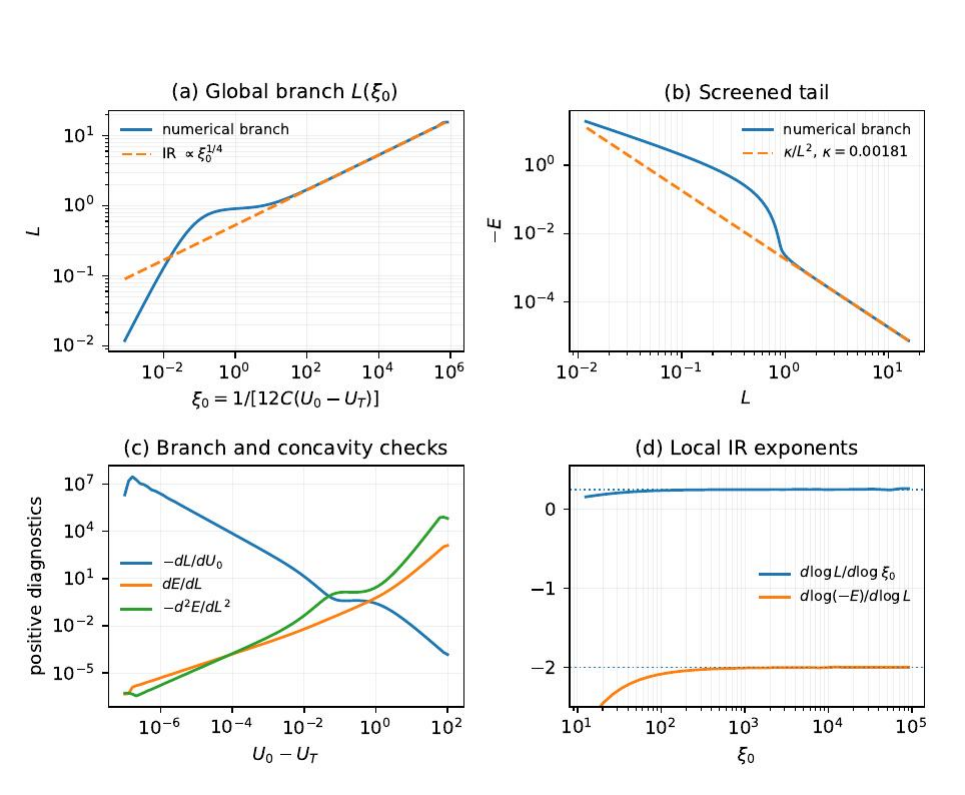}
    \caption{Numerical check of the connected branch using the exact metric factor
    $h(\xi)$.  Panel~(a) shows the global branch $L(\xi_0)$ in the smooth
    matched ansatz \eqref{eq:smooth-matched-ansatz}; at large $\xi_0$ it
    approaches the analytic $\xi_0^{1/4}$ throat scaling.  Panel~(b) shows that
    the same branch approaches the screened tail $-E=\kappa/L^2$ with
    $\kappa=\kappa_{\rm th}/144$.  Panel~(c) plots the positive quantities
    $-dL/dU_0$, $dE/dL$, and $-d^2E/dL^2$; thus the branch satisfies the
    stability/concavity diagnostics throughout the sampled range.  Panel~(d)
    shows the local exponents approaching $1/4$ and $-2$, respectively.}
    \label{fig:numerical-branch-selection}
\end{figure}

The conclusion of this numerical exercise (which used a logarithmic grid of 100
turning points over $U_0-U_T\in[10^{-7},10^2]$ with UV cutoff $U_{\max}=10^4$)
is deliberately modest but important:
within the natural smooth completion \eqref{eq:smooth-matched-ansatz}, the
screened saddle is not merely an asymptotic solution of the isolated throat.  It
lies on the same monotone, concave connected branch that emerges from the UV.
A different conclusion would require a genuine feature of the unknown
first-principles quantum-corrected matching region, not the exact throat
asymptotics alone.

The same matched-ansatz computation gives a direct observable diagnostic of the pseudo-linear
question.  Define
\begin{equation}
        n_F(L):=-\frac{d\log |F|}{d\log L},
        \qquad |F|=\frac{dE}{dL}.
\label{eq:nF-def}
\end{equation}
The Coulomb, linear, and screened regimes correspond respectively to
$n_F=2$, $n_F=0$, and $n_F=3$.  Figure~\ref{fig:money-plot} compares the exact
branch with the same matched ansatz evaluated using the
$\mathcal{O}(\xi^4)$ polynomial $h^{(4)}=1+\xi^2/60-\xi^3/126+\xi^4/240$.
For the truncated polynomial the effective wall sits at
\begin{equation}
        \xi_{\rm wall}=4.19169,
        \qquad |F|_{\rm wall}=1.54314\times10^{-2},
\label{eq:truncated-wall-force}
\end{equation}
and the force exponent tends to $n_F\to0$, producing a genuine linear asymptote
inside the truncated geometry.  In the exact geometry, by contrast, the force
passes from the Coulomb law through a crossover and then approaches
$|F|=2\kappa_{\rm RN}/L^3$; it does not remain near $n_F=0$ over an extended
interval.  The approximately linear behavior obtained from $h^{(4)}$ is
therefore a finite-order truncation shoulder, not a stable intermediate phase of
the exact matched branch.

\begin{figure}[t]
    \centering
    \includegraphics[width=0.98\textwidth]{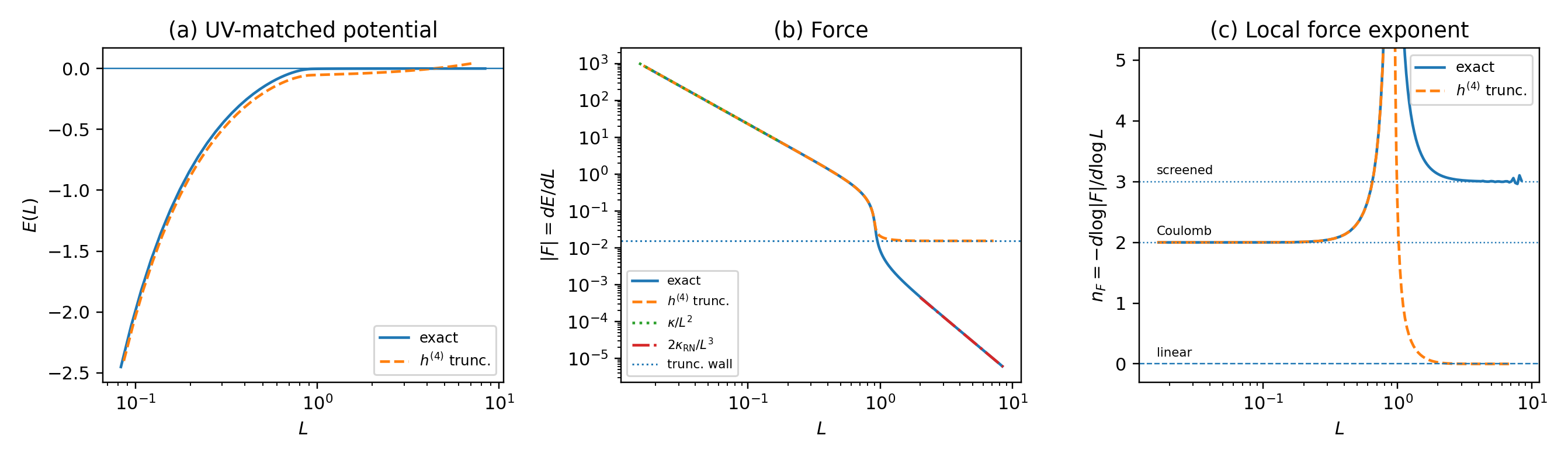}
    \caption{Observable-level comparison between the exact Schwarzian metric
    factor and the $\mathcal{O}(\xi^4)$ truncation in the smooth matched ansatz
    \eqref{eq:smooth-matched-ansatz}.  Panel~(a) shows the UV-matched static
    potential.  The truncated energy is shifted by an additive constant so that
    it agrees with the exact curve at the shortest plotted separation; the force
    and local exponent are independent of this convention.  Panel~(b) shows the
    force magnitude.  The exact branch crosses over from the Coulomb guide
    $\kappa/L^2$ to the screened guide $2\kappa_{\rm RN}/L^3$, while the
    truncated branch approaches the constant wall force in
    \eqref{eq:truncated-wall-force}.  Panel~(c) shows the local force exponent
    $n_F$.  The exact curve approaches the screened value $3$ rather than the
    linear value $0$; the constant-force plateau is a property of the finite
    polynomial extrapolation.}
    \label{fig:money-plot}
\end{figure}

\subsection{Crossover scales and modular structure}
\label{sec:crossover}

The Coulomb coefficient $\kappa$ is fixed by the AdS$_5$ asymptotics~\cite{Maldacena:1998im}, while the linear
coefficient $\sigma_{\rm rect}$ is induced by the Schwarzian/JT scale $C$~\cite{Liu:2024}.  In the smooth matched ansatz, Eq.~\eqref{eq:sigma-rect-match} gives a numerical UV calibration.  Equating the Coulombic and linear terms in \eqref{eq:E-smallL} gives
\begin{equation}
  L_c^{\rm match}
  =\sqrt{\frac{\kappa}{\sigma_{\rm rect}^{\rm match}}}\,C
  =82.4685\,C .
  \label{eq:Lc}
\end{equation}
Because this estimate uses only the first two UV terms, it should not be read as evidence for a real long plateau.  It is a matching coefficient; the observable force comparison in Fig.~\ref{fig:money-plot} tests directly whether a constant-force region persists in the exact branch.

The exact integral representation provides two natural asymptotic channels controlled by the modular pair
\begin{equation}
\tau=\frac{i\,\zeta}{\pi C}=\frac{i\,\xi}{\pi},
\qquad
q\equiv e^{\pi i\tau}=e^{-\zeta/C}=e^{-\xi},
\qquad
q^\ast\equiv e^{-\pi i/\tau}=\exp\!\Big(-\frac{\pi^2 C}{\zeta}\Big)=\exp\!\Big(-\frac{\pi^2}{\xi}\Big).
\label{eq:q-qstar}
\end{equation}
The channel weights are comparable when $|q|=|q^\ast|$, which fixes the
\emph{self-dual point} of the modular inversion $\xi\mapsto\pi^2/\xi$:
\begin{equation}
|q|=|q^\ast|
\qquad\Longleftrightarrow\qquad
\xi=\frac{\pi^2}{\xi}
\qquad\Longleftrightarrow\qquad
\xi_{\rm sd}=\pi
\quad(\text{equivalently } \zeta_{\rm sd}=\pi C).
\label{eq:zeta-selfdual}
\end{equation}
Since the turning point $\zeta_0$ increases monotonically with the boundary separation $L$,
this gives a geometric definition of a boundary crossover length,
\begin{equation}
L_{\rm sd}\;\equiv\;L(\zeta_0=\zeta_{\rm sd})=L(\zeta_0=\pi C).
\label{eq:Lsd}
\end{equation}
Unlike \eqref{eq:Lc}, which is a coefficient-matching estimate, $L_{\rm sd}$ is fixed by
the dual $q$- and $q^\ast$-expansion structure of the Mordell identity. The radial crossover
depth $\zeta_{\rm sd}=\pi C$ is intrinsic to the modular structure; the associated boundary
separation $L_{\rm sd}$, however, still depends on the $\zeta_0\mapsto L$ map and hence on
the matched geometry away from the throat.
On the connected-string branch at extremality, the modular crossover is smooth,
with no cusp or discontinuity in $E(L)$ or the force $F(L)=-dE/dL$.
In the approximately linear regime $F\approx-\sigma_{\rm eff}$, while
in the algebraic screening regime \eqref{eq:E_IR} one has $F(L)\propto -L^{-3}$;
the modular self-dual point $L_{\rm sd}$ provides a useful reference scale for
this crossover, although the actual approach to the clean $L^{-2}$ tail is broad
and is determined by the full $L(\zeta_0)$ map, as illustrated numerically in
Sec.~\ref{subsec:numerical-branch-selection}
for the minimal smooth matched ansatz.

To summarize, the single dimensionful scale $C$ from the Schwarzian sector enters the physics in three
distinct ways. First, the short-distance matching coefficient gives the UV
calibration $L_c^{\rm match}=82.4685\,C$ (Sec.~\ref{sec:short-distance}),
which should not be confused with a robust plateau scale.
Second, it controls the nonperturbative exponential corrections:
the $S$-transformed channel of the Mordell integral produces exponentially suppressed
corrections to $h(\zeta)$ of order $e^{-\pi^2 C/\zeta}$, which are invisible to any
finite Taylor truncation and represent a genuinely nonperturbative correction with respect to the $\xi\to 0$ near-boundary power expansion. Third, $C$ governs the crossover to infrared screening:
$L_{\rm sd}=L(\zeta_0=\pi C)$ provides a modular reference scale for the crossover toward the screened regime. All three
crossovers are tied to a single physical origin: the quantum fluctuations of the
Schwarzian boundary mode of JT gravity in the near-horizon AdS$_2$ region.

\subsection{Limits of finite-order truncation}
\label{subsec:finite-order-truncation}
\label{sec:truncation}

Standard perturbative analyses~\cite{Liu:2024,Liu:2024strange,Nian:2025fluid} rely on a Taylor expansion of the metric factor around the
boundary ($\zeta = 0$):
\begin{equation}
h_N(\zeta) = 1 + \sum_{k=2}^{N} a_k \left(\frac{\zeta}{C}\right)^k,
\label{eq:taylor}
\end{equation}
where the signed coefficients $a_k$ are read off from the expansion~\eqref{eq:h-exact-Taylor}
(e.g.\ $a_2=+1/60$, $a_3=-1/126$, $a_4=+1/240$).
While this series provides a valid effective description for $\zeta \ll C$, it encounters
structural limitations when extrapolated to the infrared ($\zeta \gg C$).

For any truncation at even $N>2$, the leading coefficient $a_N>0$ ensures that
$\mathcal{G}_{0,N}(\zeta)=h_N(\zeta)/\zeta^2$ diverges both as $\zeta\to 0$ and as
$\zeta\to\infty$ (since $h_N\sim \zeta^N$ at large $\zeta$). By continuity, $\mathcal{G}_{0,N}$
therefore possesses at least one global minimum at finite $\zeta$; for the $\mathcal{O}(\xi^4)$
truncation this is the stationary point at $\xi\simeq 4.19$ shown in Fig.~\ref{fig:G0_compare} (dashed curve).
However, this feature lies well outside the domain where the truncated series is accurate
(Fig.~\ref{fig:h_series}) and is thus an artifact of truncation.

By contrast, the exact solution (Eq.~\eqref{eq:h-IR-from-coth}) gives $\mathcal{G}_0(\zeta)\sim \zeta^{-3/2}$, which decays
monotonically without any extremum. Finite-order extrapolations therefore artificially
induce a spurious minimum where none exists in the exact throat geometry.

\begin{figure}[t]
  \centering
  \includegraphics[width=0.65\linewidth]{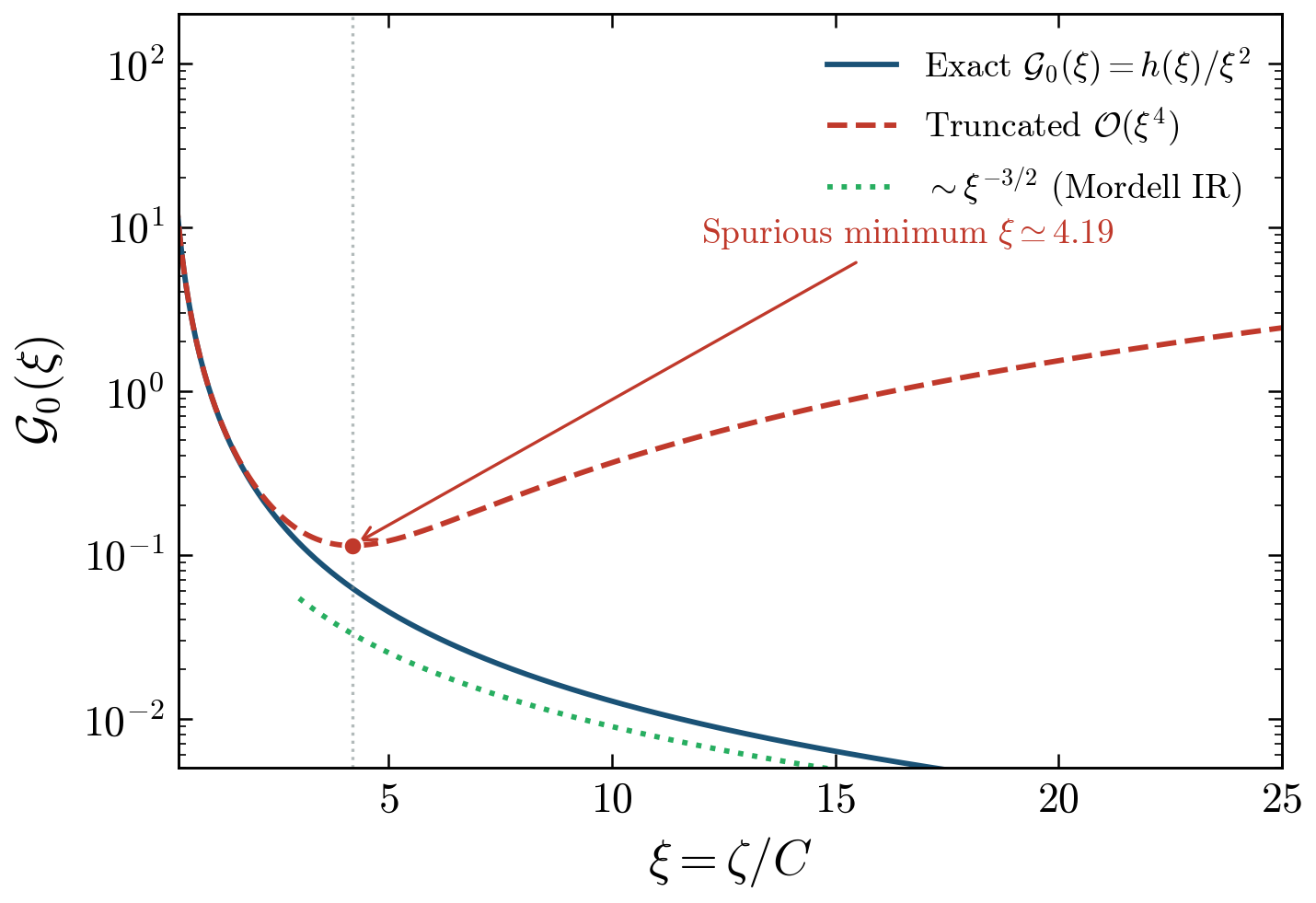}
  \caption{The ratio $\mathcal{G}_0(\xi)\equiv h(\xi)/\xi^2$ on a logarithmic scale, comparing the exact result (solid) to the $\mathcal{O}(\xi^4)$ truncation (dashed). The dotted line shows the large-$\xi$ asymptotic $\sim\xi^{-3/2}$. The truncated series develops a spurious minimum at $\xi\simeq 4.19$ (circle) followed by polynomial growth, while the exact $\mathcal{G}_0$ decays monotonically. This demonstrates that strict linear confinement is an artifact of truncation: in the exact throat geometry the rectangular Wilson-loop potential is algebraically screened, $E(L)\sim -\kappa_{\rm IR}/L^2$.}
  \label{fig:G0_compare}
\end{figure}

The exact integral representation
\begin{equation}
h(\zeta) = 2 \Bigl(\frac{\zeta}{C}\Bigr)^2 \int_0^\infty dy\, y^3 e^{-(\zeta/C)y^2} \coth(\pi y)
\end{equation}
is equivalent to a Mordell transform (Sec.~\ref{sec:mordell-eval}), evaluated explicitly in Eq.~\eqref{eq:h-q-series}, and therefore exposes two complementary expansion parameters:
$q^\ast=e^{-\pi^2 C/\zeta}$, small near the AdS$_2$ boundary ($\zeta\ll C$), and $q=e^{-\zeta/C}$, small deep in the throat ($\zeta\gg C$). In practice this
provides a controlled interpolation between the two regimes.

\begin{itemize}
\item \textit{UV expansion ($\zeta \ll C$):} The near-boundary behavior
(Sec.~\ref{sec:exact-Taylor}) contains the analytic
power series in $\zeta/C$ together with modular ($S$-transformed) terms of the form
$\sim \exp(-\pi^2 C/\zeta)$, which are invisible to any finite truncation in powers of $\zeta$.

\item \textit{IR expansion ($\zeta \gg C$):} In the deep interior the natural expansion variable
is $q=e^{-\zeta/C}$. The exact integral (Sec.~\ref{sec:IR-asymp}) gives
\begin{equation}
h(\zeta)\sim(\zeta/C)^{1/2}\left[1+\mathcal{O}\!\left(\frac{C}{\zeta}\right)\right]
+\mathcal{O}\!\left(e^{-\zeta/C}\right),
\end{equation}
so that the leading $\zeta^{1/2}$ scaling---and thus the monotone decay
$\mathcal{G}_0(\zeta)\sim \zeta^{-3/2}$---is obtained directly from the exact nonperturbative
completion rather than from any finite near-boundary truncation.
\end{itemize}

\FloatBarrier

\section{Outlook}
\label{sec:outlook}

The exact integral representation, the closed $q/q^\ast$-series evaluation,
the general-$\Delta$ extension, the variance calculation, the
complete-monotonicity/robustness results, and the observable-level force
comparison established above open several directions for further work.

For Wilson-loop applications, the main remaining problem is the global matching of $h(\zeta)$ from the AdS$_2$ throat into the full AdS$_5$ geometry. The large-$L$ screening established here depends only on the asymptotics of $h$ deep in the throat, $h\sim\sqrt{\zeta}$, and is therefore robust, but a complete description of the quark-antiquark potential at all separations requires an exact or controlled solution for the quantum-corrected warp factor across the entire radial direction. The matching ansatz of~\cite{Liu:2024} provides a starting point, but a first-principles derivation --- perhaps using the full spectral decomposition of the Schwarzian sector --- would place the intermediate-$L$ regime on the same footing as the infrared.

Within the throat itself, the fully averaged observable $\langle W[g]\rangle$ --- the Schwarzian average of the worldsheet functional, with the connected kernels of~\eqref{eq:logW-cumulants} evaluated at split arguments --- remains the natural step beyond the annealed prescription quantified in Sec.~\ref{sec:variance}.

Spatial circular Wilson loops are a natural further application.  They require a
separate disk boundary-value problem and a controlled perimeter subtraction for
the renormalized action, and are left for future work.

A closely related direction is finite temperature.  As discussed in Sec.~\ref{sec:screening}, a nonzero temperature caps the AdS$_2$ throat and should replace the extremal algebraic tail by ordinary thermal screening or disconnected-saddle physics at sufficiently large $L$.  The remaining technical problem is to carry out the analogous radial reduction of the finite-temperature kernel --- now the double-spectral representation of~\cite{Mertens:2017,Saad:2019lba}, of which the extremal kernel used here is the $\beta\to\infty$ limit --- and to track how the extremal cancellation $\Phi(0;\tau)=0$ (Eq.~\eqref{eq:Phi0-cancel}) is deformed at finite $\beta$.

On the string-theory side, the present analysis treats the string semiclassically: classical Nambu--Goto worldsheets in a quantum-corrected background. Both $\alpha'$ corrections to the worldsheet action and loop corrections from string fluctuations could modify the short-distance coefficients in Eq.~\eqref{eq:E-smallL}, and could renormalize the coefficient $\kappa_{\rm IR}$ of the infrared tail. However, the $\sqrt{\zeta}$ scaling of $h$ is a property of the Schwarzian two-point function~\cite{Mertens:2017} and is independent of the worldsheet approximation. The infrared \emph{exponent} in $E(L)\sim-\kappa_{\rm IR}/L^2$ --- and hence the qualitative screening structure --- should therefore be robust under these corrections, even if the coefficient $\kappa_{\rm IR}$ itself is not.

The mechanism identified here is not specific to the RN-AdS$_5$ Wilson-loop setup.
In backgrounds where the relevant near-horizon observable reduces to the same
Schwarzian two-point kernel~\cite{Ghosh:2019rcj}, the same Mordell-type metric factor appears;
whether the monotone decay of
the effective string tension persists depends on the behavior of the transverse
geometry in the throat.
More generally, many near-extremal observables are governed by closely related Schwarzian kernels, including transport properties~\cite{Liu:2024strange,Nian:2025fluid}, low-temperature shear correlators and $\eta/s$, with regime-dependent conclusions regarding the Kovtun--Son--Starinets bound~\cite{PandoZayas:2025etaS,Cremonini:2025etaS,Gouteraux:2025shear,Kanargias:2025shear}, Hawking radiation spectra from near-extremal rotating black holes~\cite{Maulik:2025hawking}, quasi-normal mode frequencies~\cite{Jiang:2025qnm}, and the replica phase structure of the entropy~\cite{Nian:2026replica}. The precise endpoint asymptotics and modular completions can depend on the operator dimension and observable, but the lesson is common: finite near-boundary or low-temperature truncations need not remain uniformly valid in the deep throat. Wherever a qualitative conclusion depends on extrapolating such a truncation beyond its controlled regime, an exact-kernel analysis analogous to the one performed here should be possible. When the same metric factor $h(\zeta)$ enters, the integral representation~\eqref{eq:h-exact-normalized} and the closed $q/q^\ast$-series~\eqref{eq:h-q-series} provide the required check directly.

The general-$\Delta$ subsection suggests a broader mathematical program.  We
have identified the differential ladders, the integer/half-integer modular
dichotomy, the spectral-edge law, and the dimension-dependent nonperturbative
scale $q_\Delta=e^{-\Delta^2\xi}$.  A full closed Appell--Lerch formula for
arbitrary real $\Delta$ remains open.  It would also be interesting to relate the
present Mordell/Appell--Lerch structure to the universal Schwarzian sectors in
large-$c$ two-dimensional CFTs identified in~\cite{Ghosh:2019rcj,Mertens:2017};
establishing a direct connection for the RN-AdS$_5$ Wilson-loop observable
studied here would require a separate derivation.

Finally, the explicit Mordell evaluation suggests a sharper analytic program.
Section~\ref{sec:exact-Taylor} identifies the Borel singularities of the
UV expansion at $t=-\pi^2 k^2$ and relates the first singulant ($k=1$) to the Mordell
scale $q^\ast=e^{-\pi^2/\xi}$.  The remaining resurgence problem is to describe
the analytic continuation across the negative Borel direction and its relation to
the quasi-modular $E_2$ mixing in~\eqref{eq:h-q-series}.  The closed
$q/q^\ast$ representation should also sharpen numerical error estimates at
intermediate distances (including the functional form of $\delta_{\rm mod}(L)$ in
Eq.~\eqref{eq:E-smallL}), where neither the near-boundary Taylor expansion nor the
IR power-law tail is adequate.  Deep in the throat the
$q=e^{-\zeta/C}$ corrections are exponentially small, and on the large-$L$
string branch, where $L\propto\zeta_0^{1/4}$, these become corrections of order
$\exp(-\text{const}\cdot L^4)$ to the algebraic screening tail.  Thus the
approach to $E(L)\sim -\kappa_{\rm IR}/L^2$ is power-law first, with Mordell
corrections beyond all orders of the IR $1/L$ expansion.

\section*{Acknowledgments}
\noindent We are grateful to Ricardo Esp\'indola Romero, Jorge Russo, Leonardo
Santilli and Joan Sim\'on for correspondence and useful commentary at different
stages of this work.  Part of this work was carried out while visiting the
Centro de Ciencias de Benasque Pedro Pascual, Universidad Complutense de Madrid
and Universitat de Barcelona.  We thank these institutions for the hospitality.
A talk by Jun Nian at the International Congress of Basic Science (ICBS) in
Beijing in 2025 provided motivation for some of the developments in this work.

\FloatBarrier

\appendix
\numberwithin{equation}{section}

\section{Numerical implementation details}
\label{app:numerics}

This appendix collects the computational specifications underlying
Figs.~\ref{fig:h_series}--\ref{fig:G0_compare}, in sufficient detail for
independent reproduction.  All calculations were performed in double-precision
floating-point arithmetic ($\approx 15$--$16$ significant digits).

\subsection{\texorpdfstring{Evaluation of $h(\xi)$ and $\mathcal{G}_0(\xi)$}{Evaluation of h(xi) and G0(xi)}}
\label{app:h-eval}

The exact metric factor~\eqref{eq:h-exact-normalized},
\begin{equation}
h(\xi)=2\xi^{2}\int_0^\infty dy\,y^3\,e^{-\xi y^2}\,\coth(\pi y)\,,
\label{eq:h-quadrature}
\end{equation}
was evaluated by adaptive Gauss-Kronrod quadrature (the \texttt{scipy.integrate.quad} routine, which wraps the QUADPACK library~\cite{Piessens:1983}).
The integrand is smooth on $(0,\infty)$; the factor $\coth(\pi y)$ has a $1/y$ singularity
at $y=0$, but the product $y^3\coth(\pi y)\sim y^2/\pi$ is regular.
For numerical stability, the integral was split at $y=1$:
on $[0,1]$, we subtracted the $1/(\pi y)$ piece of $\coth(\pi y)$ and evaluated the singular part $\int_0^1 y^2 e^{-\xi y^2}/\pi\,dy$ analytically
(an incomplete gamma function).  The remainder and the tail $[1,\infty)$ were computed
by adaptive quadrature.  Absolute and relative tolerances were set to $10^{-12}$.

\medskip\noindent\textit{Convergence check.}
Increasing the tolerances to $10^{-14}$ changed $h(\xi)$ by less than $10^{-13}$ at all sampled
$\xi\in[10^{-3},10^5]$.  The truncation of the upper integration limit was checked by verifying
that the Gaussian factor $e^{-\xi y^2}$ renders the tail contribution below $10^{-15}$ for
$y>y_{\max}$ with $y_{\max}=\max(20,\,5/\sqrt{\xi})$.

The confinement indicator $\mathcal{G}_0(\xi)=h(\xi)/\xi^2$ was obtained by dividing the
same integral, and its derivatives (used in checking complete monotonicity numerically) were
computed by evaluating the corresponding moment integrals
$(-1)^n\mathcal{G}_0^{(n)}(\xi)=2\int_0^\infty y^{3+2n}e^{-\xi y^2}\coth(\pi y)\,dy$
with the same quadrature scheme.

\medskip\noindent\textit{Variance evaluation.}
The variance \eqref{eq:var-exact} was evaluated with the same split quadrature
applied to the moments $\int_0^\infty dy\,y^{3+2n}e^{-\xi y^{2}}\coth(\pi y)$,
$n=0,1,2$; the identity between \eqref{eq:h2-coth} and the
$\mathcal{G}_0$-derivative form in \eqref{eq:var-exact} was confirmed to
$12$ digits over $0.5\le\xi\le10$, and the infrared coefficient
$\sqrt{\pi}/3$ to four digits at $\xi=10^{4}$.

\subsection{Evaluation of the Wilson-loop parametric integrals}
\label{app:rect-numerics}

The parametric integrals~\eqref{eq:numerical-L-throat}--\eqref{eq:numerical-E-throat}
(throat-only) and the full matched-ansatz
integrals~\eqref{eq:L-parametric}--\eqref{eq:E-parametric} were evaluated as follows.

\medskip\noindent\textit{Endpoint regularization.}
Both integrals develop an integrable square-root singularity at $u\to 1^-$ (throat)
or $v\to 1^+$ (full geometry).  This was removed by the substitution $u=1-s^2$
(equivalently $v=1+s^2$), which renders the transformed integrand bounded.  The
regularized integrals were computed by adaptive Gauss-Kronrod quadrature with
absolute tolerance $10^{-10}$ and relative tolerance $10^{-10}$.

\medskip\noindent\textit{Turning-point grid.}
For the throat-only computation, 200 values of $\xi_0$ were spaced logarithmically
over $[10^{-1},10^5]$.  For the matched-ansatz computation, 100 values of
$U_0-U_T$ were spaced logarithmically over $[10^{-7},10^2]$.  At each grid point,
the function $h(\xi)$ was evaluated by the quadrature scheme of
Appendix~\ref{app:h-eval}; there is no interpolation table.

\medskip\noindent\textit{UV cutoff.}
The full-geometry integrals were cut off at $U_{\max}=10^4$.  The cutoff dependence
was checked by repeating the computation at $U_{\max}=5\times 10^3$ and $U_{\max}=5\times 10^4$:
the relative change in the renormalized energy was below $10^{-6}$ at all sampled
turning points.

\medskip\noindent\textit{Derivative diagnostics.}
The stability and concavity conditions~\eqref{eq:numerical-sign-checks} involve first
and second derivatives of the parametric curve $(L(U_0),E(U_0))$.  These were
computed by centered finite differences on the $(L,E)$ data:
\begin{equation}
\frac{dE}{dL}\bigg|_i
\simeq\frac{E_{i+1}-E_{i-1}}{L_{i+1}-L_{i-1}}\,,
\qquad
\frac{d^2E}{dL^2}\bigg|_i
\simeq\frac{(E_{i+1}-E_i)/(L_{i+1}-L_i)-(E_i-E_{i-1})/(L_i-L_{i-1})}{(L_{i+1}-L_{i-1})/2}\,.
\label{eq:finite-diff}
\end{equation}
The logarithmic grid spacing ensures that the fractional step
$\Delta\log(U_0-U_T)\approx 0.09$ is small enough for the finite-difference
quotients to be accurate; doubling the number of grid points to 200 changed the
computed signs and local exponents by less than $0.1\%$.

\medskip\noindent\textit{Observable force comparison.}
For Fig.~\ref{fig:money-plot}, the force was evaluated directly from the first
integral,
\begin{equation}
        |F|(U_0)=\frac{U_0^2}{2\pi}\sqrt{f(U_0)h(U_0)},
\label{eq:force-numerics}
\end{equation}
in the same units $U_T=R=C=1$ used in the matched-ansatz computation.  The exact branch used
150 logarithmically spaced values of $U_0-U_T\in[10^{-6},10^3]$.  The truncated
branch used $h^{(4)}$ and was sampled as
$U_0-U_T=(U_{\rm wall}-U_T)+s$, with
$s\in[10^{-7},10^{2.5}]$ and $U_{\rm wall}-U_T=0.0198806$.  The local exponent
$n_F$ was computed by centered finite differences on $(\log L,\log |F|)$.  For
this plotting-only comparison, $h(\xi)$ was tabulated on 2000 logarithmically
spaced points in $\xi\in[10^{-8},10^8]$ using the equivalent quadrature
\begin{equation}
        h(\xi)=\int_0^\infty dx\;x e^{-x}
        \coth\!\left(\pi\sqrt{x/\xi}\right),
\label{eq:h-fast-quadrature}
\end{equation}
and interpolated with a monotone piecewise cubic interpolant~\cite{Fritsch:1980} in $(\log\xi,\log h)$; direct
quadrature spot checks changed the plotted force by less than the line width.
The exact energy in Fig.~\ref{fig:money-plot}(a) was reconstructed by integrating
$dE/dL=|F|$ and fixing the largest-$L$ endpoint to the analytic tail
$E=-\kappa_{\rm RN}/L^2$.  The truncated energy was shifted by an additive
constant to agree with the exact curve at the shortest plotted separation.

\medskip\noindent\textit{Numerical values.}
Table~\ref{tab:rect-checks} collects the key numerical outputs and their
sensitivity to the computational parameters.

\begin{table}[h]
\centering
\begin{tabular}{lccc}
\hline
Quantity & Value & $U_{\max}$ sensitivity & Grid sensitivity \\
\hline
$(-E_{\rm th})L_{\rm th}^2$ ($\xi_0>10^3$) & $0.26003$ & $<10^{-5}$ & $<10^{-4}$ \\
$d\log(-E_{\rm th})/d\log L_{\rm th}$ & $-2.001$ & $<10^{-3}$ & $<10^{-3}$ \\
$(-E)L^2$ (matched, $\xi_0>10^4$) & $1.804\times 10^{-3}$ & $<10^{-6}$ & $<10^{-4}$ \\
$d\log(-E)/d\log L$ (matched) & $-2.003$ & $<10^{-3}$ & $<10^{-3}$ \\
Analytic $\kappa_{\rm th}$ & $0.259921$ & --- & --- \\
Analytic $\kappa_{\rm RN}=\kappa_{\rm th}/144$ & $1.80501\times 10^{-3}$ & --- & --- \\
\hline
\end{tabular}
\caption{Key numerical outputs from the rectangular Wilson-loop evaluation and their
sensitivity to computational parameters.  ``$U_{\max}$ sensitivity'' is the
relative change when $U_{\max}$ is varied by a factor of 5.  ``Grid sensitivity''
is the relative change when the number of turning points is doubled.}
\label{tab:rect-checks}
\end{table}

\medskip\noindent\textit{Software.}
All computations used Python~3.11 with NumPy~1.26 and SciPy~1.12.
Plots were generated with Matplotlib~3.8.  The source bundle includes
\texttt{code/reproduce\_figures.py}, which exposes the quadrature for $h(\xi)$,
$\mathcal{V}(\xi)$, the universal throat curve, and the matched force
comparison, and reproduces Figs.~\ref{fig:variance},
\ref{fig:universal-throat}, and~\ref{fig:money-plot}.

\FloatBarrier

\bibliographystyle{unsrt}
\bibliography{refs}

\end{document}